\documentclass[prodmode,acmjacm]{acmsmall}
\usepackage{graphicx}
\usepackage{amsmath}
\usepackage{tikz}
\usetikzlibrary{arrows,positioning,automata}
\usepackage{multirow}
\usepackage{hyperref}

\allowdisplaybreaks[1]

\newcommand{\st}{\:|\:}
\newcommand{\N}{\mathbb{N}}
\newcommand{\Q}{\mathbb{Q}}
\newcommand{\R}{\mathbb{R}}

\renewcommand{\P}{\mathsf{P}}
\newcommand{\FP}{\mathsf{FP}}
\newcommand{\NP}{\ensuremath{\mathsf{NP}}}
\newcommand{\PH}{\ensuremath{\mathsf{PH}}}
\newcommand{\PP}{\ensuremath{\mathsf{PP}}}
\newcommand{\sharpP}{\ensuremath{\mathsf{\# P}}}
\newcommand{\PSPACE}{\ensuremath{\mathsf{PSPACE}}}
\newcommand{\EXPTIME}{\ensuremath{\mathsf{EXPTIME}}}

\newcommand{\hard}{\mathcal{H}}
\newcommand{\soft}{\mathcal{S}}
\newcommand{\improvs}{I}
\newcommand{\valids}{A}
\newcommand{\wref}{w_{\text{ref}}}
\newcommand{\eopt}{\epsilon_\mathrm{opt}}

\newcommand{\cic}[2]{\textsc{CI}\textup{(}#1\textup{,}#2\textup{)}}
\newcommand{\mcic}[2]{\textsc{MCI}\textup{(}#1\textup{,}#2\textup{)}}
\newcommand{\SAT}{SAT}
\newcommand{\sharpSAT}{\#\SAT}
\newcommand{\NSAT}{\#NSAT}
\newcommand{\ucfgint}{\textsc{\#UCFG-Int}}
\newcommand{\bpcp}{\textsc{Bounded-PCP}}
\newcommand{\nbpcp}{\#\bpcp}

\newcommand{\DFA}{\textup{DFA}}
\newcommand{\NFA}{\textup{NFA}}
\newcommand{\CFG}{\textup{CFG}}
\newcommand{\UCFG}{\textup{UCFG}}
\newcommand{\symb}{\textsc{Symb}}

\title{Control Improvisation}
\author{DANIEL J. FREMONT, ALEXANDRE DONZ\'E, and SANJIT A. SESHIA \affil{\\University of California, Berkeley}}



\begin{abstract} 
We formalize and analyze a new problem in formal language theory termed \emph{control improvisation}.
Given a specification language, the problem is to produce an \emph{improviser}, a probabilistic algorithm that randomly generates words in the language, subject to two additional constraints: the satisfaction of a quantitative \emph{soft constraint}, and the exhibition of a specified amount of randomness.
Control improvisation has many applications, including for example systematically generating random test vectors satisfying format constraints or preconditions while being similar to a library of seed inputs.
Other applications include robotic surveillance, machine improvisation of music, and randomized variants of the supervisory control problem.
We describe a general framework for solving the control improvisation problem, and use it to give efficient algorithms for several practical classes of instances with finite automaton and context-free grammar specifications.
We also provide a detailed complexity analysis, establishing $\sharpP$-hardness of the problem in many other cases.
For these intractable cases, we show how symbolic techniques based on Boolean satisfiability ({\SAT}) solvers can be used to find approximate solutions.
Finally, we discuss an extension of control improvisation to multiple soft constraints that is useful in some applications.
\end{abstract} 

\begin{document}

\begin{bottomstuff}
This work is supported in part by the National Science Foundation Graduate Research Fellowship Program under Grant No. DGE-1106400, by NSF grants CCF-1139138 and CNS-1646208, and by TerraSwarm, one of six centers of STARnet, a Semiconductor Research Corporation program sponsored by MARCO and DARPA.\smallskip

A preliminary version of this article, coauthored with the late David Wessel, appeared in the \emph{Proceedings of the 35th IARCS Annual Conference on Foundations of Software Technology and Theoretical Computer Science (FSTTCS 2015)}. \smallskip

Authors' addresses: D. J. Fremont, Group in Logic and the Methodology of Science, UC Berkeley, Berkeley, CA 94720, email: dfremont@berkeley.edu; A. Donz\'e, Department of Electrical Engineering and Computer Sciences, UC Berkeley, Berkeley, CA 94720, email: donze@berkeley.edu; S. A. Seshia, Department of Electrical Engineering and Computer Sciences, UC Berkeley, Berkeley, CA 94720, email: sseshia@eecs.berkeley.edu. \smallskip
\end{bottomstuff}

\maketitle

\section{Introduction}

We introduce and formally characterize a new formal language-theoretic problem termed {\em control
improvisation}. Given a specification language or \emph{hard constraint}, the problem is to produce an \emph{improviser}, a probabilistic
algorithm that randomly generates words in the language, subject to two additional
constraints: the generated words should satisfy a quantitative \emph{soft constraint}, and the improviser must exhibit a specified amount of randomness.

The original motivation for this problem came from a topic known as {\em machine improvisation of
music}~\cite{rowe-2001}.  Here, the goal is to create algorithms which can generate variations of a
reference melody like those commonly improvised by human performers, for example in jazz.  Such an
algorithm should have three key properties.  First, the melodies it generates should conform to rhythmic
and melodic constraints typifying the music style (e.g. in jazz, the melodies should follow the
harmonic conventions of that genre). Second, the generated
melodies should be actual variations on the reference melody, neither reproducing it exactly nor being
so different as to be unrecognizable.
Finally, the algorithm should be sufficiently randomized that
running it several times produces a variety of different improvisations.
In previous work~\cite{donze-icmc14}, we identified these
properties in an initial definition of the control improvisation problem, and applied it to the
generation of monophonic (solo) melodies over a given jazz song harmonization\footnote{Examples of
improvised melodies can be found at the following URL:\\
\url{http://www.eecs.berkeley.edu/~donze/impro_page.html}.}.
 
These three properties of a generation algorithm are not specific to music.  Consider
\emph{black-box fuzz testing} \cite{fuzzing-book}, which produces many inputs to a program hoping to
trigger a bug.  Often, constraints are imposed on the generated inputs, e.g. in \emph{generative}
fuzz testing approaches which enforce an appropriate format so that the input is not rejected
immediately by a parser.  Also common are \emph{mutational} approaches which guide the generation
process with a set of real-world seed inputs, generating only inputs which are variations of those
in the set. And of course, fuzzers use randomness to ensure that a variety of inputs are tried.
Thus we see that the inputs generated in fuzz testing have the same general requirements as music
improvisations: satisfying a set of constraints, being appropriately similar/dissimilar to a
reference, and being sufficiently diverse.

We propose control improvisation as a precisely-defined theoretical problem capturing these
requirements, which are common not just to the two examples above but to many other generation
problems.  Potential applications also include home automation mimicking typical occupant behavior (e.g.,
randomized lighting control obeying time-of-day constraints and limits on energy
usage~\cite{akkaya-iotdi}) and randomized variants of the supervisory control problem
\cite{lafortune06}, where a controller keeps the behavior of a system within a safe operating region
(the language of an automaton) while adding diversity to its behavior via randomness.  A typical example
of the latter is surveillance: the path of a patrolling robot should satisfy various constraints
(e.g. not running into obstacles) and be similar to a predefined route, but incorporate some
randomness so that its location is not too predictable \cite{lafortune-personal15}.

Our focus, in this paper, is on the {\em theoretical characterization of
control improvisation}. Specifically, we give a precise theoretical
definition and a rigorous characterization of the complexity of the
control improvisation problem as a function of the type of specification used for the hard and soft constraints. While the problem is distinct from any other we have
encountered in the literature, our methods are closely connected to
prior work on random sampling from the languages of automata and
grammars \cite{hickey-cohen,sharpNFA}, and sampling from
the satisfying assignments of a Boolean formula~\cite{jvv,bgp,unigen}.
Probabilistic programming techniques~\cite{probprog} could be used
for sampling under constraints, but the present methods cannot be used
to construct improvisers meeting our definition.

In summary, this paper makes the following novel contributions:
\begin{itemize}
\item Formal definitions of the notions of control improvisation (CI) and a polynomial-time improvisation scheme (Sec.~\ref{sec:definition}); \smallskip
\item A theoretical characterization of the conditions under which improvisers exist, and a general framework for constructing efficient improvisation schemes (Sec.~\ref{sec:existence}); \smallskip
\item A detailed complexity analysis of several practical classes of CI instances with specifications of different kinds: \smallskip
    \begin{itemize}
    \item \emph{Finite automata}: a polynomial-time improvisation scheme for deterministic automata and a $\sharpP$-hardness result for nondeterministic automata (Sec.~\ref{sec:automata}); \smallskip
    \item \emph{Context-free grammars:} polynomial-time improvisation schemes when one specification is an unambiguous grammar, and a $\sharpP$-hardness result when both specifications are such grammars (Sec.~\ref{sec:grammars}); \smallskip
    \item \emph{Boolean formulas:} a symbolic approach based on Boolean satisfiability (SAT) solving that can be used, for example, with finite-state automata that are too large to represent explicitly (Sec.~\ref{sec:symbolic}); \smallskip
    \end{itemize}
\item An extension of the basic control improvisation definition to allow multiple soft constraints (Sec.~\ref{sec:multiple-soft}).
\end{itemize}

Compared to the earlier version of this paper (\citeN{fsttcs-version}), the CI definition has been elaborated with explicit length bounds and a more general randomness requirement.
The material on generic improvisation schemes, context-free grammar specifications, and multiple soft constraints is entirely new.

\section{Background and Problem Definition} \label{sec:definition}

In this section, we first define notation and provide background on a previous
automata-theoretic method for music improvisation based on a data
structure called the {\em factor oracle}. We then provide a formal
definition of the control improvisation problem while explaining the
choices made in the definition.

\subsection{Notation}

We abbreviate deterministic and nondeterministic finite
automata as DFAs and NFAs respectively.
Similarly we refer to a context-free grammar as a CFG, and an unambiguous context-free grammar as a UCFG.
We write the set $\{1, \dots, k\}$ for $k \in \N^+$ as $[k]$.
For any alphabet $\Sigma$, we denote the length of a finite word $w \in \Sigma^*$ by $|w|$.
We also write $\Sigma^{\le n}$ for $\cup_{0 \le i \le n} \Sigma^i$.
Finally, if $L$ is a language over $\Sigma$, we sometimes denote $\Sigma^* \setminus L$ by $\overline{L}$.

We refer to several standard complexity classes: $\P$ and $\NP$, the class $\sharpP$ of counting versions of $\NP$ problems, its decision analogue $\PP$, the polynomial-time hierarchy $\PH$, and the class $\FP$ of polynomial-time computable functions.
For definitions of these classes and the relationships between them, see \citeN{arora-barak}.

\subsection{Motivation: Factor Oracles}

An effective and practical approach to machine improvisation of music (used for example in the prominent OMax system \cite{omax}) is based on a data structure
called the factor oracle \cite{AssayagD04,Cleophas03constructingfactor}. Given a word
$\wref$ of length $N$ that is a symbolic encoding of a reference melody, a factor oracle $F$ is an
automaton constructed from $\wref$ with the following key properties: $F$ has $N+1$ states, all
accepting, chained linearly with direct transitions labelled with the letters in $\wref$, and with
potentially additional forward and backward transitions. Figure \ref{figure:factor-oracle} depicts
$F$ for $\wref = bbac$.  A word $w$ accepted by $F$ consists of concatenated ``factors'' (subwords) of $\wref$,
and its dissimilarity with $\wref$ is correlated with the number of non-direct transitions. By
assigning a small probability $\alpha$ to non-direct transitions, $F$ becomes a generative Markov
model with tunable ``divergence'' from $\wref$. In order to impose more musical structure on the
generated words, our previous work~\cite{donze-icmc14} additionally requires that improvisations
satisfy rules encoded as deterministic finite automata, by taking the product of the generative
Markov model and the DFAs. While this approach is heuristic and lacks any formal guarantees, it has
the basic elements common to machine improvisation schemes: (i) it involves randomly generating
strings from a formal language, (ii) it includes a quantitative requirement on which strings are admissible based on their divergence from a reference string, and (iii) it enforces diversity in the
generated strings. The definition we propose below captures these elements in a
rigorous theoretical manner, suitable for further analysis.  In Sec.~\ref{sec:automata}, we
revisit the factor oracle, sketching how the notion of divergence from $\wref$ that it represents can be
encoded in our formalism.
 
 {
\setlength{\intextsep}{8pt}
\setlength{\belowcaptionskip}{-5pt}
\setlength{\abovecaptionskip}{0pt}
\begin{figure}[tp]
\centering
\begin{tikzpicture}[initial text=, transform shape, scale=0.8]

 \node[accepting, state, initial] (s0) {$s_0$}; 
 \node[accepting, state, right= of s0] (s1) {$s_1$}; 
 \node[accepting, state, right= of s1] (s2) {$s_2$};
 \node[accepting, state, right= of s2] (s3) {$s_3$}; 
 \node[accepting, state, right= of s3] (s4) {$s_4$};

 \path[->] 
 (s0) edge node [above] {$b$} (s1)
 (s1) edge node [above] {$b$} (s2)    
 (s2) edge node [above] {$a$} (s3)
 (s3) edge node [above] {$c$} (s4) 
 (s0) edge [bend left=40] node [above] {$a$} (s3)
 (s0) edge [bend left=50] node [above] {$c$} (s4)
 (s1) edge [bend left] node [above] {$a$} (s3)
 (s1) edge [bend left] node [above] {$\epsilon$} (s0) 
 (s2) edge [bend left ] node [above] {$\epsilon$} (s1) 
 (s3) edge [bend left] node [above] {$\epsilon$} (s0) 
 (s4) edge [bend left] node [above] {$\epsilon$} (s0); 

\end{tikzpicture}
\caption{Factor oracle constructed from the word $\wref = bbac$.}
\label{figure:factor-oracle}
\end{figure}
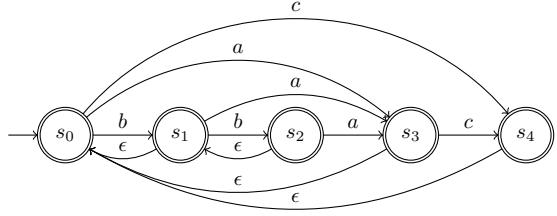
}
 
\subsection{Problem Definition}

A control improvisation problem is defined in terms of two languages specifying the hard and soft constraints.
The complexity of the problem depends on how these languages are represented, and we will consider several representations in this paper.
We use the general term \emph{specification} to refer to the representation of a language.
\begin{definition}
A \emph{specification} $\mathcal{X}$ is a finite representation of a language $L(\mathcal{X})$ over a finite alphabet $\Sigma$.
For example, $\mathcal{X}$ could be a finite automaton or a context-free grammar.
\end{definition}

\begin{definition} \label{defn:improvs}
Fix a \emph{hard specification} $\hard$, a \emph{soft specification} $\soft$, and length bounds $m, n \in \N$.
An \emph{improvisation} is any word $w \in L(\hard)$ such that $m \le |w| \le n$, and we write $\improvs$ for the set of all improvisations.
An improvisation $w \in \improvs$ is \emph{admissible} if $w \in L(\soft)$, and we write $\valids$ for the set of all admissible improvisations.
\end{definition}

Observe that because we impose an upper bound on the length of an improvisation, the sets $\improvs$ and $\valids$ are finite\footnote{Note that the earlier version of this paper (\citeN{fsttcs-version}) did not include a length bound as part of the problem definition, and so some of the complexity results here are different.}. \medskip

\textit{Running Example.}
Our concepts will be illustrated with a simple 
example. Our aim is to produce variations of the binary string $s =
001$ of length 3, subject to the constraint that there cannot be two
consecutive 1s. So we put $\Sigma = \{0,1\}$ and $m = n = 3$, and let 
$\hard$ be a DFA which accepts all strings 
that do not have two 1s in a row. To ensure that our
variations are similar to $s$, we let $\soft$ be a DFA accepting words with Hamming distance at most 1 from $s$. Then the improvisations are the strings
$000$, $001$, $010$, $100$, and $101$, of which $000$, $001$, and
$101$ are admissible. \medskip

Intuitively, an improviser samples from the set of improvisations according to some
distribution. But what requirements must one impose on this distribution?
First, we require our generated improvisation to
be admissible with probability at least $1-\epsilon$ for some specified $\epsilon$. When the soft specification $\soft$ encodes a notion of similarity to a reference string, for example, this allows us to
require that our improvisations usually be similar to the reference.
Note that this requirement can be relaxed by increasing the value of $\epsilon$, justifying the term ``soft specification''.
Other types of quantitative requirements are interesting and worthy of study, but in this paper we focus on the probabilistic constraint above.
Note also that while the soft specification $\soft$ by itself is a property true or false for any word, the soft constraint requirement is \emph{not} a property of individual words: it constrains the \emph{distribution} from which we generate words.
This is an important asymmetry between the hard and soft constraints: in particular, it means that we must consider both the language $L(\soft)$ and its \emph{complement} $\overline{L(\soft)}$, as we may have to generate words from each.

Second, since we want a variety of improvisations, we impose a \emph{randomness} requirement: every improvisation must be generated with a probability within a desired range $[\lambda, \rho]$.
This requirement is designed to accommodate two different needs for randomness.
In fuzz testing, for example, in order to ensure coverage we might want to ensure that every test case can be generated and so put $\lambda > 0$.
By contrast, in music improvisation it is not important that every possible melody can be generated; rather, we simply want no improvisation to arise too frequently.
So in this application we might put $\lambda = 0$ but $\rho < 1$.
Our definition allows any combination of these two cases.
Note that if there are $N$ improvisations, then setting $\lambda$ or $\rho$ equal to $1/N$ forces the distribution to be uniform.
So uniform sampling as used in fuzz testing or constrained random verification is a special case of control improvisation.
We also note that other randomness requirements are possible, for example ensuring variety by imposing some minimum distance between generated improvisations.
This could be reasonable in a setting such as music or robotics where there is a natural metric on the space of improvisations, but we choose to keep our setting general and not assume such a metric.
Combining these requirements,
we obtain our definitions of an acceptable distribution over improvisations and thus of an
improviser:

\begin{definition} \label{defn:feasible}
Given $\mathcal{C} = (\hard, \soft, m, n, \epsilon, \lambda, \rho)$
with $\hard$, $\soft$, $m$, and $n$ as in
Definition~\ref{defn:improvs}, $\epsilon \in [0,1]
\cap \Q$ an error probability, and $\lambda, \rho \in [0,1] \cap \Q$
probability 
bounds, a distribution $D : \Sigma^*
\rightarrow [0,1]$ is an
\emph{improvising distribution} if it satisfies the following requirements:
\begin{describe}{\emph{Hard constraint:}}
\item[\emph{Hard constraint:}] $\Pr [w \in \improvs \st w \leftarrow D ] = 1$ \smallskip
\item[\emph{Soft constraint:}] $\Pr[ w \in \valids \st w \leftarrow D ] \ge 1 - \epsilon$ \smallskip
\item[\emph{Randomness:}] $\forall w \in I$, $\; \lambda \le D(w) \le \rho$
\end{describe}
If there is an improvising distribution, we say that
$\mathcal{C}$ is \emph{feasible}. An 
\emph{improviser} for a feasible $\mathcal{C}$ 
is an expected finite-time
probabilistic algorithm generating strings in $\Sigma^*$
whose output distribution (on empty
input) is an improvising distribution.
\end{definition}

To summarize, if $\mathcal{C}$ is feasible, there
exists a distribution satisfying the requirements in Definition~\ref{defn:feasible}, and an
improviser is a probabilistic algorithm for sampling from one. \\

\textit{Running Example.}
For our running example, $\mathcal{C} = (\hard, \soft,
3, 3, 0, 0, 1/4)$ is not feasible since $\epsilon=0$ 
means we can only generate admissible improvisations, and since there
are only 3 of those we cannot possibly give them all probability at
most $1/4$. Increasing $\rho$ to $1/3$ would make $\mathcal{C}$
feasible. Increasing $\epsilon$ to $1/4$ would also work, allowing us
to return an inadmissible improvisation $1/4$ of the time:
an algorithm uniformly sampling from $\{000, 001, 101, 100\}$ would be
an improviser for $(\hard, \soft, 3, 3, 1/4, 0, 1/4)$.
This would not be an improviser if we put $\lambda > 0$, however, since the improvisation $010$ is generated with probability $0$.

\begin{definition}
Given $\mathcal{C} = (\hard, \soft, m, n, \epsilon, \lambda, \rho)$,
the \emph{control improvisation (CI)} problem is to decide whether
$\mathcal{C}$ is feasible, and if so to generate an improviser for
$\mathcal{C}$.
The \emph{size} $|\mathcal{C}|$ of a CI instance is the total size of the bit representation of its parameters, treating $m$ and $n$ as being represented in unary and all other numerical parameters in binary.
\end{definition}

When more than one CI instance is being discussed, we use subscripts to disambiguate parameters: for example $n_\mathcal{C}$ is the value of $n$ in the instance $\mathcal{C}$.
Similarly, we write $\improvs_\mathcal{C}$ and $\valids_\mathcal{C}$ for the sets of improvisations and admissible improvisations respectively with respect to $\mathcal{C}$.

Since we are interested in the complexity of the CI problem as a function of the type of specification used, we define classes of instances based on those:
\begin{definition}
If $\mathcal{A}$ and $\mathcal{B}$ are classes of specifications, $\cic{\mathcal{A}}{\mathcal{B}}$ is the class of CI instances $\mathcal{C} = (\hard, \soft, m, n, \epsilon, \lambda, \rho)$ where $\hard \in \mathcal{A}$ and $\soft \in \mathcal{B}$.
When discussing decision problems, we use the same notation for the feasibility problem associated with the class (i.e. given $\mathcal{C} \in \cic{\mathcal{A}}{\mathcal{B}}$, decide whether it is feasible).
\end{definition}
For example, $\cic{\UCFG}{\DFA}$ is the class of instances where the hard specification is an unambiguous context-free grammar and the soft specification is a deterministic finite automaton.

Ideally, we would like an efficient algorithm to solve the CI problem. Furthermore,
the improvisers our algorithm produces should themselves be efficient, in the sense
that their runtimes are polynomial in the size of the original CI
instance. This leads to our last definition:
\begin{definition} \label{defn:scheme}
A \emph{polynomial-time improvisation scheme} for a class
$\mathcal{P}$ of CI instances is an algorithm $S$
with the following properties: 
\begin{describe}{\emph{Improviser efficiency:}}
\item[\emph{Correctness:}] For any $\mathcal{C} \in \mathcal{P}$, if $\mathcal{C}$ is
feasible then $S(\mathcal{C})$ is an improviser for $\mathcal{C}$, and
otherwise $S(\mathcal{C}) = \bot$. \smallskip
\item[\emph{Scheme efficiency:}] There is a polynomial $p : \R \rightarrow \R$ such that the runtime of $S$ on any $\mathcal{C} \in \mathcal{P}$ is at most $p(|\mathcal{C}|)$. \smallskip
\item[\emph{Improviser efficiency:}] There is a polynomial $q : \R \rightarrow \R$ such that for every $\mathcal{C} \in \mathcal{P}$, if $G = S(\mathcal{C}) \ne \bot$ then $G$ has expected runtime at most $q(|\mathcal{C}|)$. 
\end{describe}
\end{definition}
The improviser efficiency requirement prevents the scheme from, for example, being polynomial-time only by offloading an exponential search into the generated improviser.
Accordingly, when we say an improvisation scheme is polynomial-time relative to an oracle, the runtimes of both the scheme and its generated improvisers are measured relative to the oracle.

A polynomial-time improvisation scheme for a class of CI instances is
an efficient, uniform way to solve the control improvisation problem
for that class. In subsequent sections, we will identify several classes that have
such improvisation schemes.

\section{Existence of Improvisers, and a Generic Improvisation Scheme} \label{sec:existence}

In this section, we first give necessary and sufficient conditions for an improviser to exist for a given CI instance.
The proof is constructive, and leads to a generic procedure for solving the CI problem using a few abstract operations on specifications.
All of the improvisation schemes given in this paper are derived by instantiating this procedure for particular classes of CI instances.
Finally, we use the procedure to give a general $\sharpP$ upper bound on the complexity of CI for a wide range of specifications.

First, existence of improvisers.
It turns out that the feasibility of an improvisation problem is completely determined by the sizes of $\improvs$ and $\valids$:
\begin{theorem} \label{theorem:feasibility}
For any $\mathcal{C} = (\hard, \soft, m, n, \epsilon, \lambda, \rho)$, the following are equivalent:
\begin{enumerate}
\item $\mathcal{C}$ is feasible.
\item The following inequalities hold\footnote{If $\lambda = 0$, we treat division by zero as yielding $\infty$, so that both the inequalities involving $\lambda$ are trivially satisfied. This makes sense, as $\lambda = 0$ means we do not impose a lower bound on the probabilities of individual improvisations. The remaining inequalities are precisely those given in the corresponding theorem in our earlier paper \cite[Theorem 3.1]{fsttcs-version}, which did not allow a lower bound in the randomness requirement.}:
 \begin{enumerate}
 \item $1/\rho \le |\improvs| \le 1/\lambda$
 \item $(1-\epsilon)/\rho \le |\valids|$
 \item $|\improvs| - |\valids| \le \epsilon/\lambda$
 \end{enumerate}
\item There is an improviser for $\mathcal{C}$.
\end{enumerate}
\end{theorem}
\begin{proof}
\begin{describe}{(1)$\Rightarrow$(2):}
\item[(1)$\Rightarrow$(2):] Suppose $D$ is an improvising distribution.
Then $\rho |\improvs| = \sum_{w \in \improvs} \rho \ge \sum_{w \in \improvs} D(w) = \Pr[w \in \improvs \st w \leftarrow D] = 1$, so $|\improvs| \ge 1/\rho$.
Similarly, $\lambda |\improvs| = \sum_{w \in \improvs} \lambda \le \sum_{w \in \improvs} D(w) = \Pr[w \in \improvs \st w \leftarrow D] = 1$, so $|\improvs| \le 1/\lambda$.
Since $A \subseteq I$, we also have $\rho |\valids| = \sum_{w \in \valids} \rho \ge \sum_{w \in \valids} D(w) = \Pr[ w \in \valids \st w \leftarrow D] \ge 1 - \epsilon$; therefore $|\valids| \ge (1-\epsilon)/\rho$.
Finally, we have $\lambda |\improvs \setminus \valids| = \sum_{w \in \improvs \setminus \valids} \lambda \le \sum_{w \in \improvs \setminus \valids} D(w) = \Pr[w \in \improvs \setminus \valids \st w \leftarrow D] = \Pr[w \in \improvs \st w \leftarrow D] - \Pr[w \in \valids \st w \leftarrow D] \le 1 - (1 - \epsilon) = \epsilon$, so $|\improvs| - |\valids| \le \epsilon/\lambda$.

\item[(2)$\Rightarrow$(3):] Putting $\eopt = \max ( 1 - \rho |\valids|, \lambda |\improvs \setminus \valids| )$, since $|\improvs \setminus \valids| \le |\improvs| \le 1/\lambda$ we have $0 \le \eopt \le 1$.
Let $D$ be the distribution on $\improvs$ which picks uniformly from $\valids$ with probability $1 - \eopt$ and otherwise picks uniformly from $\improvs \setminus \valids$.
Note that this distribution is well-defined, since if $\valids = \emptyset$ then $\eopt = 1$, and if $\improvs \setminus \valids = \emptyset$ then $\rho |\valids| = \rho |\improvs| \ge 1$ and so $\eopt = 0$.
Clearly, $D$ satisfies the hard constraint requirement.

Now if $\eopt = 1 - \rho |\valids|$ we have $\eopt \le 1 - \rho \cdot (1-\epsilon)/\rho = \epsilon$.
Otherwise, $\eopt = \lambda |\improvs \setminus \valids| \le \lambda (\epsilon / \lambda) = \epsilon$.
So in either case we have $\Pr[ w \in \valids \st w \leftarrow D] = 1 - \eopt \ge 1 - \epsilon$, and $D$ satisfies the soft constraint requirement.

For any $w \in \valids$, we have $D(w) = (1 - \eopt) / |\valids| \le (1 - (1 - \rho |\valids|)) / |\valids| = \rho$.
If $\eopt = 1 - \rho |\valids|$, then $D(w) = \rho \ge \lambda$; otherwise $\eopt = \lambda |\improvs \setminus \valids|$, so $D(w) = (1 - \lambda |\improvs \setminus \valids|) / |\valids| = (1 - \lambda |\improvs| + \lambda |\valids|) / |\valids| \ge (1 - 1 + \lambda |\valids|) / |\valids| = \lambda$.
Thus $D(w) \ge \lambda$ in either case.
Similarly, for any $w \in \improvs \setminus \valids$ we have $D(w) = \eopt / |\improvs \setminus \valids| \ge \lambda |\improvs \setminus \valids| / |\improvs \setminus \valids| = \lambda$.
If $\eopt = 1 - \rho |\valids|$, then $D(w) = (1 - \rho |\valids|) / |\improvs \setminus \valids| = (1 - \rho |\valids|) / (|\improvs| - |\valids|) \le (1 - \rho |\valids|) / ((1/\rho) - |\valids|) = \rho$; otherwise $\eopt = \lambda |\improvs \setminus \valids|$, so $D(w) = (\lambda |\improvs \setminus \valids|) / |\improvs \setminus \valids| = \lambda \le \rho$.
Therefore for any $w \in \improvs$ we always have $\lambda \le D(w) \le \rho$, and thus $D$ satisfies the randomness requirement.

This shows that $D$ is an improvising distribution.
Since it has finite support and rational probabilities, there is an expected finite-time probabilistic algorithm sampling from it, and this algorithm is an improviser for $\mathcal{C}$.

\item[(3)$\Rightarrow$(1):] Immediate. \qed
\end{describe} 
\end{proof}

\begin{remark}
In fact, whenever $\mathcal{C}$ is feasible, the construction in the proof of Theorem \ref{theorem:feasibility} gives an improviser which works in nearly the most trivial possible way: it has two finite lists $S$ and $T$, flips a (biased) coin to decide which list to use, and then returns an element of that list uniformly at random.
There can of course be other improvising distributions assigning a variety of probabilities between $\lambda$ and $\rho$, but one of this simple form is always guaranteed to exist.
\end{remark}

The intuition behind the improviser constructed in Theorem \ref{theorem:feasibility} is as follows: attempt to assign the minimum allowed probability $\lambda$ to the elements of $\improvs \setminus \valids$, spreading the rest uniformly over $\valids$.
This may assign a probability to the elements of $\valids$ that is greater than the maximum allowed probability $\rho$.
In that case, we clamp the probabilities of the admissible improvisations at $\rho$, and spread the remainder uniformly over $\improvs \setminus \valids$.
This intuition suggests that the error probability $\eopt$ achieved by this improviser is optimal, since we attempt to give $\valids$ as high a probability as we can.
In fact, this is the case (justifying the notation $\eopt$): the smallest error probability compatible with the hard constraint and randomness requirements is $\eopt$.

\begin{corollary}
$\mathcal{C}$ is feasible if and only if $1/\rho \le |\improvs| \le 1/\lambda$ and $\epsilon \ge \eopt$, where $\eopt = \max(1 - \rho |\valids|, \lambda |\improvs \setminus \valids|)$.
\end{corollary}
\begin{proof}
Immediate from Theorem \ref{theorem:feasibility}, noting that $\epsilon \ge \eopt$ if and only if inequalities (2b) and (2c) in the Theorem hold.
\end{proof}

Beyond just giving conditions for feasibility, the proof of Theorem \ref{theorem:feasibility} suggests a generic procedure for solving the CI problem.
It uses a few basic operations on specifications:
\begin{describe}{\emph{Uniform Sampling:}}
\item[\emph{Intersection:}] Given two specifications $\mathcal{X}$ and $\mathcal{Y}$, compute a specification $\mathcal{Z}$ such that $L(\mathcal{Z}) = L(\mathcal{X}) \cap L(\mathcal{Y})$. \smallskip
\item[\emph{Difference:}] Given two specifications $\mathcal{X}$ and $\mathcal{Y}$, compute a specification $\mathcal{Z}$ such that $L(\mathcal{Z}) = L(\mathcal{X}) \setminus L(\mathcal{Y})$. \smallskip
\item[\emph{Length Restriction:}] Given a specification $\mathcal{X}$ and $m,n \in \N$ in unary, compute a specification $\mathcal{Y}$ such that $L(\mathcal{Y}) = \{ w \in L(\mathcal{X}) \st m \le |w| \le n \}$. \smallskip
\item[\emph{Counting:}] Given a specification $\mathcal{X}$ and a bound $n \in \N$ in unary on the length of strings in $L(\mathcal{X})$ (which is therefore finite), compute $|L(\mathcal{X})|$. \smallskip
\item[\emph{Uniform Sampling:}] Given a specification $\mathcal{X}$ and a bound $n \in \N$ in unary on the length of strings in $L(\mathcal{X})$, sample uniformly at random from $L(\mathcal{X})$.
\end{describe}

If these operations can be implemented efficiently for a particular type of specification, then we can efficiently solve the corresponding CI problems:
\begin{theorem} \label{theorem:generic-scheme}
Suppose \textsc{Spec} is a class of specifications that supports the operations above.
Suppose further that the operations can be done in polynomial time (expected time in the case of uniform sampling).
Then there is a polynomial-time improvisation scheme for $\cic{\textsc{Spec}}{\textsc{Spec}}$.
\end{theorem}
\begin{proof}
Given a problem $\mathcal{C} \in \cic{\textsc{Spec}}{\textsc{Spec}}$ with $\mathcal{C} = (\hard, \soft, m, n, \epsilon, \lambda, \rho)$, the scheme works as follows.
First, we construct representations of all the sets we need.
Applying length restriction to $\hard$, we obtain a specification $\mathcal{\improvs}$ such that $L(\mathcal{\improvs}) = \{ w \in L(\hard) \st m \le |w| \le n \} = \improvs$.
Applying intersection to $\mathcal{\improvs}$ and $\soft$, we get a specification $\mathcal{\valids}$ such that $L(\mathcal{\valids}) = L(\mathcal{\improvs}) \cap L(\soft) = \improvs \cap L(\soft) = \valids$.
Finally, applying difference to $\mathcal{\improvs}$ and $\mathcal{\valids}$ gives a specification $\mathcal{B}$ such that $L(\mathcal{B}) = \improvs \setminus \valids$.

Next, applying counting to $\mathcal{\improvs}$ and $\mathcal{\valids}$ (with bound $n$, since all words in $\improvs$ have length at most $n$), we compute $|\improvs|$ and $|\valids|$.
We can then check whether the inequalities in Theorem \ref{theorem:feasibility} are satisfied.
If not, then $\mathcal{C}$ is not feasible and we return $\bot$.
Otherwise, $\mathcal{C}$ is feasible and we will build an improviser sampling from the same distribution constructed in the proof of Theorem \ref{theorem:feasibility}.
As there, let $\eopt = \max(1 - \rho |\valids|, \lambda |\improvs \setminus \valids|)$ (easily computed since we know $|\improvs|$ and $|\valids|$, and $|\improvs \setminus \valids| = |\improvs| - |\valids|$), and let $D$ be the distribution on $\improvs$ that with probability $1 - \eopt$ picks uniformly from $\valids$ and otherwise picks uniformly from $\improvs \setminus \valids$.
Since the inequalities in Theorem \ref{theorem:feasibility} are true, its proof shows that $D$ is an improvising distribution for $\mathcal{C}$.

We can easily build a probabilistic algorithm $G$ sampling from $D$: it simply flips a coin, applying uniform sampling to $\mathcal{\valids}$ with probability $1 - \eopt$ and otherwise applying uniform sampling to $\mathcal{B}$.
Since $L(\mathcal{\valids}) = \valids$ and $L(\mathcal{B}) = \improvs \setminus \valids$, the output distribution of $G$ is $D$ and so $G$ is an improviser for $\mathcal{C}$.

This procedure is clearly correct.
Since the intersection, difference, and length restriction operations take polynomial time, the constructed specifications $\mathcal{\improvs}$, $\mathcal{\valids}$, and $\mathcal{B}$ have sizes polynomial in $|\hard|$, $|\soft|$, $m$, and $n$, and thus in $|\mathcal{C}|$.
So the subsequent counting and sampling operations performed on these specifications will also be polynomial in $|\mathcal{C}|$.
In particular, the computed values of $|\improvs|$ and $|\valids|$ have polynomial bitwidth, so the same is true of $\eopt$, and the arithmetic performed by the procedure also takes time polynomial in $|\mathcal{C}|$.
Therefore in total the procedure runs in polynomial time.
The improvisers generated by the procedure run in expected polynomial time, since the bitwidth of $\eopt$ is polynomial and uniform sampling takes expected polynomial time.
So the procedure is a polynomial-time improvisation scheme.
\end{proof}

\textit{Running Example.}
Recall that for our running example $\mathcal{C} = (\hard, \soft, 3, 3, 1/4, 0, 1/4)$, we have $\improvs = \{000, 001, 010, 100, 101\}$ and $\valids = \{000, 001, 101\}$.
In Sec.~\ref{sec:automata} we will show that all the operations needed by Theorem \ref{theorem:generic-scheme} can be performed efficiently for DFAs.
Applying length restriction, intersection, and difference, we obtain DFAs $\mathcal{\improvs}$, $\mathcal{\valids}$, and $\mathcal{B}$ such that $L(\mathcal{\improvs}) = \improvs$, $L(\mathcal{\valids}) = \valids$, and $L(\mathcal{B}) = \improvs \setminus \valids = \{ 010, 100 \}$.
Applying counting we find that $|\improvs| = 5$ and $|\valids| = 3$, so the inequalities in Theorem \ref{theorem:feasibility} are satisfied.
Next we compute $\eopt = \max(1 - \rho |\valids|, \lambda |\improvs \setminus \valids|) = \max(1/4, 0) = 1/4$.
Finally, we return an improviser $G$ that samples uniformly from $L(\mathcal{\valids})$ with probability $1 - \eopt = 3/4$ and from $L(\mathcal{B})$ with probability $\eopt = 1/4$.
So $G$ returns $000$, $001$, and $101$ with probability $3/4 \cdot 1/3 = 1/4$ each, and $010$ and $100$ with probability $1/4 \cdot 1/2 = 1/8$ each.
This distribution satisfies our conditions, so it is indeed an improviser for $\mathcal{C}$. \medskip

Note that since we use the operations on specifications as black boxes, the construction above relativizes: if the operations can be done in polynomial time relative to an oracle $\mathcal{O}$, then the Theorem yields an improvisation scheme that is polynomial-time relative to $\mathcal{O}$ (where both the scheme and the generated improvisers may query $\mathcal{O}$).
This allows us to upper bound the complexity of CI problems under very weak assumptions on the type of specifications.

\begin{theorem} \label{theorem:sharpP-bound}
Suppose $\textsc{Spec}$ is a class of specifications such that membership in the language of a specification can be decided in polynomial time relative to a $\PH$ oracle.
Then $\cic{\textsc{Spec}}{\textsc{Spec}}$ has an improvisation scheme that is polynomial-time relative to a $\sharpP$ oracle.
\end{theorem}
Without the $\PH$ oracle, this essentially follows from the fact that if specifications are decision algorithms with uniformly polynomial runtimes which test membership in their respective languages, then counting and sampling can be done with a $\sharpP$ oracle.
However, instantiating the general operations used by Theorem \ref{theorem:generic-scheme} is somewhat tricky since we need to keep track of the runtimes of the specifications.
So instead we just verify that the particular counting and sampling queries needed by the construction in Theorem \ref{theorem:generic-scheme} can be implemented with a $\sharpP$ oracle.
This whole argument relativizes to the $\PH$ oracle, which can then be eliminated via Toda's theorem.
\begin{proof}[of Theorem \ref{theorem:sharpP-bound}]
By assumption there is a polynomial-time algorithm $M^{\PH}(w, \mathcal{X})$ that decides whether $w \in L(\mathcal{X})$ for any specification $\mathcal{X} \in \textsc{Spec}$.
Then $\improvs_\mathcal{C} = \{ w \in L(\hard_\mathcal{C}) \st m_\mathcal{C} \le |w| \le n_\mathcal{C} \} = \{ w \in \Sigma^* \st m_\mathcal{C} \le |w| \le n_\mathcal{C} \land M^{\PH}(w, \hard_\mathcal{C}) = 1 \}$, so whether $w \in \improvs_\mathcal{C}$ is decidable in time polynomial in $|w|$ and $|\mathcal{C}|$ relative to a $\PH$ oracle.
Similarly, $\valids_\mathcal{C} = \improvs_\mathcal{C} \cap L(\soft_\mathcal{C}) = \improvs_\mathcal{C} \cap \{ w \in \Sigma^* \st M^{\PH}(w, \soft_\mathcal{C}) = 1\}$, so whether $w \in \valids_\mathcal{C}$ is decidable in time polynomial in $|w|$ and $|\mathcal{C}|$ relative to a $\PH$ oracle.
Since any $w \in \improvs_\mathcal{C}$ has length at most $n_\mathcal{C}$, and in particular polynomial in $|\mathcal{C}|$, this shows that the relations $R_\improvs = \{ (\mathcal{C}, w) \st w \in \improvs_\mathcal{C} \}$, $R_\valids = \{ (\mathcal{C}, w) \st w \in \valids_\mathcal{C} \}$, and $R_{\improvs \setminus \valids} = \{ (\mathcal{C}, w) \st w \in \improvs_\mathcal{C} \setminus \valids_\mathcal{C} \}$ are $\NP^{\PH}$-relations\footnote{In the terminology of \citeN{bgp}, relativized to the $\PH$ oracle; these relations (unrelativized) are called $p$-relations by \citeN{jvv}.}.
Therefore computing $|\improvs_\mathcal{C}| = |\{ w \in \Sigma^* \st (\mathcal{C}, w) \in R_\improvs \}|$ is a $\sharpP^{\PH}$ problem, and likewise for $|\valids_\mathcal{C}|$.
Furthermore, sampling uniformly from $\valids_\mathcal{C}$ and $\improvs_\mathcal{C} \setminus \valids_\mathcal{C}$ is equivalent to sampling uniformly from the witnesses of $\mathcal{C}$ under the relations $R_\valids$ and $R_{\improvs \setminus \valids}$ respectively.
This can be done in polynomial expected time relative to a $\sharpP^{\PH}$ oracle (relativizing the algorithms of \citeN{jvv} or \citeN{bgp}).
So the construction in Theorem \ref{theorem:generic-scheme} yields an improvisation scheme for $\cic{\textsc{Spec}}{\textsc{Spec}}$ that is polynomial-time relative to a $\sharpP^{\PH}$ oracle.

To remove the $\PH$ oracle, note that $\PP^\PH \subseteq \P^\PP$ by the stronger form of Toda's theorem \cite{toda}.
Since $\sharpP^\PH \subseteq \FP^{\PP^\PH}$ by a standard binary search argument, we have $\sharpP^\PH \subseteq \FP^{\PP^\PH} \subseteq \FP^{\P^\PP} = \FP^\PP \subseteq \FP^\sharpP$.
So the polynomial-time improvisation scheme (and the improvisers it generates) can simulate the $\sharpP^\PH$ oracle using only a $\sharpP$ oracle.
\end{proof}

This general argument establishes that for a wide range of specifications (i.e. those which do not require resources beyond $\PH$ to evaluate), the complexity of the CI problem is between polynomial-time and $\sharpP$.
In the next few sections, we will pin down the complexity of several natural classes of CI problems as being at one end or the other of this range.

\section{Automata Specifications} \label{sec:automata}

In this section we consider specifications that are finite automata.
For deterministic automata, we give a polynomial-time improvisation scheme, but show that for nondeterministic automata the CI problem becomes $\sharpP$-hard.

DFAs are perhaps the simplest type of specification.
They capture the notion of
a \emph{finite-memory} specification, where membership of
a word can be determined by scanning the word left-to-right, only
being able to remember a finite number of already-seen symbols. An
example of a finite-memory specification $\soft$ is one such that
$w \in L(\soft)$ iff each subword of $w$
of a fixed constant length satisfies some condition.

\begin{example}[Factor Oracles]
Recall that one way of measuring the divergence of an improvisation
$w$ generated by the factor oracle $F$ built from a word $\wref$ is by
counting the number of non-direct transitions that $w$ causes $F$ to
take. Since DFAs cannot do unbounded counting, we can use a sliding window of some
finite size $k$. Then our soft specification $\soft$ can be that
at any point as $F$ processes $w$, the number of the previous $k$
transitions which were non-direct lies in some interval $[\ell,h]$. This predicate can be encoded as a DFA
of size $O(|F| \cdot 2^k)$ as follows: we have a copy of $F$, denoted $F_s$, for every string $s \in \{0,1\}^k$, each bit of $s$ indicating whether the corresponding previous transition (out of the last $k$) was non-direct.
As each new symbol is processed, we execute the current copy of $F$ as usual, but move to the appropriate state of the copy of $F$ corresponding to the new $k$-transition history, i.e., if we were in $F_s$, we move to $F_t$ where $t$ consists of the last $k-1$ bits of $s$ followed by a 0 if the transition we took was direct and a 1 otherwise.
Making the states of $F_s$ accepting iff the number of 1s in $s$ is in $[\ell,h]$, this automaton represents $\alpha$ as desired.
The size of
the automaton grows exponentially in the size of the window, but for
small windows it can be reasonable.
\end{example}

When both the hard and soft specifications are DFAs, we can instantiate the procedure of Theorem \ref{theorem:generic-scheme} to get a polynomial-time improvisation scheme.
To do this, we use a classical method for counting and uniformly sampling from the language of a DFA (see for example \citeN{hickey-cohen}).
The next two lemmas summarize the results we need, proofs being given in the Appendix for completeness.

\begin{lemma} \label{lemma:dfa-counting}
If $\mathcal{D}$ is a DFA, $|L(\mathcal{D})|$ can be computed in polynomial time.
\end{lemma}

\begin{lemma} \label{lemma:dfa-sampling}
There is a probabilistic algorithm that given a DFA $\mathcal{D}$ with finite language, returns a uniform random sample from $L(\mathcal{D})$ in polynomial expected time.
\end{lemma}

Using these techniques, we have the following:
\begin{theorem} \label{theorem:dfa-scheme}
There is a polynomial-time improvisation scheme for \cic{\DFA}{\DFA}.
\end{theorem}
\begin{proof}
We instantiate the five operations needed by Theorem \ref{theorem:generic-scheme}.
\begin{describe}{\emph{Uniform Sampling:}}
\item[\emph{Intersection:}] Given two DFAs $\mathcal{X}$ and $\mathcal{Y}$, we can construct a DFA $\mathcal{Z}$ such that $L(\mathcal{Z}) = L(\mathcal{X}) \cap L(\mathcal{Y})$ with the standard product construction.
The time needed to do this and the size of $\mathcal{Z}$ are both $O(|\mathcal{X}| |\mathcal{Y}|)$. \smallskip
\item[\emph{Difference:}] Given two DFAs $\mathcal{X}$ and $\mathcal{Y}$, we can construct a DFA $\mathcal{Z}$ such that $L(\mathcal{Z}) = L(\mathcal{X}) \setminus L(\mathcal{Y})$ by complementing $\mathcal{Y}$ with the standard construction, and then taking the product with $\mathcal{X}$.
The time required and resulting automaton size are polynomial, as for intersection. \smallskip
\item[\emph{Length Restriction:}] Given a DFA $\mathcal{X}$ and $m,n \in \N$ in unary, we can construct a DFA $\mathcal{R}$ of size $O(n)$ such that $L(\mathcal{R}) = \{ w \in \Sigma^* \st m \le |w| \le n \}$ (the states track how many symbols have been read so far, up to $n+1$).
Then letting $\mathcal{Y}$ be the product of $\mathcal{X}$ and $\mathcal{R}$, we have $L(\mathcal{Y}) = \{ w \in L(\mathcal{X}) \st m \le |w| \le n \}$ as desired.
Polynomial runtime and automaton size follow from the same properties for the product operation. \smallskip
\item[\emph{Counting:}] Lemma \ref{lemma:dfa-counting}. \smallskip
\item[\emph{Uniform Sampling:}] Lemma \ref{lemma:dfa-sampling}.
\end{describe}
Since we can perform all of these operations in polynomial time, Theorem \ref{theorem:generic-scheme} yields a polynomial-time improvisation scheme.
\end{proof}

So we can efficiently solve CI problems with specifications given by deterministic finite automata.
However, if we have \emph{nondeterministic} automata, this technique breaks down because there is no longer a one-to-one correspondence between accepting paths in the automaton and words in its language.
Indeed, counting the language of an NFA is known to be hard, and this translates into hardness of CI.
\begin{theorem} \label{theorem:nfa-hardness}
\cic{\NFA}{\DFA} and \cic{\DFA}{\NFA} are $\sharpP$-hard.
\end{theorem}
\begin{proof}
We prove this for \cic{\NFA}{\DFA} --- the other case is similar. As shown by \citeN{sharpNFA}, the problem of determining $| L(\mathcal{M}) \cap \Sigma^m |$ given an NFA $\mathcal{M}$ over an alphabet $\Sigma$ and $m \in \N$ in unary is $\sharpP$-complete. We give a polynomial-time (Cook) reduction from this problem to checking feasibility of a CI instance in \cic{\NFA}{\DFA}.

For any NFA $\mathcal{M}$ and positive $N \in \N$, consider the \cic{\NFA}{\DFA} instance $\mathcal{C}_N = (\mathcal{M}, \mathcal{T}, m, m, 0, 0, 1/N)$ where $\mathcal{T}$ is the trivial DFA accepting all of $\Sigma^*$.
Clearly for this instance we have $\improvs = \valids = L(\mathcal{M}) \cap \Sigma^m$.
By Theorem \ref{theorem:feasibility}, $\mathcal{C}_N$ is feasible iff $|L(\mathcal{M}) \cap \Sigma^m| \ge 1 / (1/N) = N$, and so we can determine whether the latter is the case with a feasibility query for a \cic{\NFA}{\DFA} instance.
Since $|L(\mathcal{M}) \cap \Sigma^m| \le |\Sigma|^m$, using binary search we can find the exact value of $|L(\mathcal{M}) \cap \Sigma^m|$ with polynomially-many such queries (recalling that $m$ is given in unary).
\end{proof}

This result indicates that in general, the control improvisation
problem is probably intractable in the presence of NFAs or more complex automata such as probabilistic automata\footnote{The earlier version of this paper (\citeN{fsttcs-version}) did not impose a length bound as part of the CI problem, so that the problem was actually undecidable for probabilistic automata. Under the current definition, the problem is solvable with a $\sharpP$ oracle by Theorem \ref{theorem:sharpP-bound}.} \cite{rabin-pfas}.

\section{Context-Free Grammar Specifications} \label{sec:grammars}

Next we consider specifications that are context-free grammars.
These are useful, for example, in the fuzz testing application where we want to generate tests conforming to a file format.
Unfortunately, CFGs are more concise than NFAs, and so our hardness result for the latter immediately rules out a polynomial-time improvisation scheme.
\begin{theorem} \label{theorem:cfg-hardness}
\cic{\CFG}{\DFA} and \cic{\DFA}{\CFG} are $\sharpP$-hard.
\end{theorem}
\begin{proof}
Follows from Theorem \ref{theorem:nfa-hardness}, since an NFA can be converted to an equivalent CFG in polynomial time.
\end{proof}

However, CFGs used in practice are often designed to be \emph{unambiguous}, enabling faster parsing.
For UCFGs there are polynomial-time algorithms for counting and uniform sampling \cite{hickey-cohen,mckenzie}, giving hope for a polynomial-time improvisation scheme.
A complication is that UCFGs are not closed under intersection, and indeed we will see below that when both hard and soft specifications are UCFGs and intersecting them is required, the CI problem is $\sharpP$-hard.
But when one specification is a UCFG and the other is a DFA, there is a polynomial-time improvisation scheme.

To use the procedure of Theorem \ref{theorem:generic-scheme} in this case, we need to be able to intersect a UCFG $G$ and a DFA $D$.
By a classical result of \citeN{bar-hillel}, their intersection is a context-free language, and we can compute a CFG $H$ such that $L(H) = L(G) \cap L(D)$ in polynomial time.
However, in order to then sample from this intersection, we need $H$ to be unambiguous.
As noted by \citeN{ginsburg-ullian}, this is actually accomplished by the construction of \citeN{bar-hillel}.
However their presentation does not explicitly demonstrate this fact, so for completeness we present a proof.
We also modify the algorithm in several ways to improve its complexity (this is important in practice since the construction, while polynomial, has relatively high degree).

\begin{lemma} \label{lemma:regular-intersection}
Given a UCFG $G$ and a DFA $D = (Q, \Sigma, \delta, q_0, F)$ over a common alphabet $\Sigma$, there is an algorithm which computes a UCFG $H$ such that $L(H) = L(G) \cap L(D)$ in $O(|G| |D|^3)$ time.
\end{lemma}
\begin{proof}
First convert $G$ to a CFG $G' = (V, \Sigma, P, S)$ such that
\begin{enumerate}
 \item $L(G') = L(G) \setminus \{\epsilon\}$, and
 \item the RHS of every production in $P$ has length at most $2$ and does not contain $\epsilon$.
\end{enumerate}
This transformation can be done in a way that the time required and the size of $G'$ are both $O(|G|)$, and in the process we determine whether $\epsilon \in L(G)$ \cite[Sec.~7.4.2]{hmu}.
Furthermore, it is simple to check that since $G$ is unambiguous, this procedure ensures that $G'$ is also.
The rest of our algorithm will build an unambiguous CFG $H$ such that $L(H) = L(G') \cap L(D)$, and thus $\epsilon \not \in L(H)$.
If $\epsilon \in L(G) \cap L(D)$ (which we can easily check, since we know whether $\epsilon \in L(G)$ from above, and checking if $\epsilon \in L(D)$ is trivial), we can simply add the production $S \rightarrow \epsilon$ to $H$ without introducing any ambiguity.
Therefore ultimately we will have $L(H) = L(G) \cap L(D)$ as desired.

The main construction of \citeN{bar-hillel} works on DFAs with a single accepting state: there is an initial preprocessing step which writes $D$ as a union of $|F|$ DFAs of that form, so that the construction can be carried out on each and the resulting grammars combined.
Doing this would contribute a factor of $|F|$ to our algorithm's runtime, so instead we modify $D$ to produce an NFA $D' = (Q', \Sigma, \delta', q_0, F')$ as follows.
We add two new states $\textsc{Accept}$ and $\textsc{Reject}$, where $\delta(\textsc{Accept}, a) = \delta(\textsc{Reject}, a) = \textsc{Reject}$ for every $a \in \Sigma$.
Also for any transition $\delta(x, a) = y$ where $y \in F$, we add a transition from $x$ to $\textsc{Accept}$ on input $a$ --- note that this makes $D'$ nondeterministic.
Finally, we make $\textsc{Accept}$ the only accepting state of $D'$.
Clearly we can construct $D'$ in time linear in $|D|$.

Now consider any nonempty word $w \in L(D)$, which we may write $a_0 \dots a_n$ for some $n \ge 0$.
Let $w'$ be $w$ without its last symbol $a_n$ (so $w' = \epsilon$ if $w$ has length 1).
Since $w \in L(D)$ there is some $x \in Q$ and $y \in F$ such that $\delta(q_0, w') = x$ and $\delta(x, a_n) = y$.
Therefore $w \in L(D')$, since $D'$ on input $w'$ can follow unmodified transitions from $D$ to reach $x$, and then the new transition from $x$ to $\textsc{Accept}$ on input $a_n$.
Conversely, if $w \in L(D')$ then executing $D'$ on input $w$ we must end in state $\textsc{Accept}$, and all transitions except the last must be transitions of $D$ (since once we follow a new transition to $\textsc{Accept}$, any further transitions will end in $\textsc{Reject}$).
Note that this means there is only one accepting path in $D'$ corresponding to $w$ --- this will be important later.
If the last transition in this path is from state $x$ on input $a$, then by construction $\delta(x, a) \in F$ and so $w \in L(D)$.
Therefore $L(D') = L(D) \setminus \{\epsilon\}$, and since $\epsilon \not \in L(G')$ we have $L(G') \cap L(D) = L(G') \cap L(D')$.

Now, following \citeN{bar-hillel}, we build the CFG $H = (\hat{V}, \Sigma, \hat{P}, \hat{S})$ where $\hat{V} = (Q' \times V \times Q') \cup \Sigma$ and $\hat{S} = (q_0, S, \textsc{Accept})$.
There are two\footnote{In fact there are four additional kinds, where the RHS is of the form $bc$, $bC$, $Bc$, or $B$. But these are easily handled along the same lines as types (1) and (2) above, so for simplicity we omit the details. The reason for dealing with all of these different types of productions instead of converting $G$ to Chomsky normal form (which only has productions of types (1) and (2)) is that the conversion to CNF can square the size of the grammar. The conversion we use only expands the grammar by a linear amount at most, while keeping the construction of $H$ efficient.} kinds of productions in $\hat{P}$:
\begin{enumerate}
\item \label{variable-productions} For every production $A \rightarrow BC$ in $P$, we add the productions $(x, A, z) \rightarrow (x, B, y) (y, C, z)$ for all $x,y,z \in Q'$ to $\hat{P}$.
\item \label{terminal-productions} For every production $A \rightarrow b$ in $P$, we add the productions $(x, A, y) \rightarrow b$ for all $x,y \in Q'$ such that $(x, b, y) \in \delta'$ to $\hat{P}$.
\end{enumerate}
For a proof that $L(H) = L(G') \cap L(D')$ (and thus that $L(H) = L(G') \cap L(D)$, as desired), see \citeN{bar-hillel} (note however that our type \ref{terminal-productions} productions are split into two separate productions in their presentation).
Clearly, we can construct $H$ in $O(|P| \cdot |Q'|^3) = O(|G| |D|^3)$ time.

It remains to show that $H$ is unambiguous.
Take any word $w \in L(H)$, and any two derivation trees $T_1$ and $T_2$ for $w$.
Observe that by the way we constructed the productions of $H$ above, if we drop the first and third components of the labels on the non-leaf nodes of $T_1$ or $T_2$ we obtain a derivation tree for $w$ from the grammar $G'$.
Therefore since $G'$ is unambiguous, $T_1$ and $T_2$ have the same tree structure.

Now we prove by induction on trees that every node in $T_1$ has the property that its label is identical to the corresponding node in $T_2$.
For the leaf nodes this is immediate, since they spell out $w$ in either tree.
For nodes one level up, i.e. those to which a production of type \ref{terminal-productions} is applied, it is proved by \citeN{bar-hillel} that if these are written from left to right they form a sequence $(q_0, v_0, q_1), (q_1, v_1, q_2), \dots, (q_{n-1}, v_{n-1}, q_n)$ where $q_0 q_1 \dots q_n$ is an accepting path for $w$ in $D'$.
Since as shown above there is only one such path, these nodes also have the property.
Finally, observe that under productions of type \ref{variable-productions}, the state labels of the parent node are uniquely determined by those of its children.
So if all of a node's children satisfy the property, so does the node itself, and thus by induction all nodes have the property.
Therefore $T_1$ and $T_2$ are identical, and so $H$ is unambiguous.
\end{proof}

The next two lemmas put the known results on counting and sampling for UCFGs into the form we need.
\begin{lemma} \label{lemma:ucfg-counting}
Given a UCFG $G$ and $m,n \in \N$ in unary, we can compute $|\{ w \in L(G) \st m \le |w| \le n \}|$ in polynomial time.
\end{lemma}
\begin{proof}
The algorithms of \citeN{hickey-cohen} or \citeN{mckenzie} can compute $|L(G) \cap \Sigma^i|$ in time polynomial in $|G|$ and $i$ (both algorithms do this as a preliminary step to sampling uniformly from $L(G) \cap \Sigma^i$).
We do this for every $i$ such that $m \le i \le n$, and return the sum of the results\footnote{In practice, the calls for each $i$ should not be done independently, but share memoization tables.}.
\end{proof}

\begin{lemma} \label{lemma:ucfg-sampling}
Given a UCFG $G$ and $m,n \in \N$ in unary, we can uniformly sample from $\{ w \in L(G) \st m \le |w| \le n \}$ in expected polynomial time.
\end{lemma}
\begin{proof}
As in the previous lemma, compute for each $i$ from $m$ to $n$ the count $c_i = |L(G) \cap \Sigma^i|$.
Pick a random $j$ so that the probability of obtaining $i$ is $c_i / \sum_k c_k$.
Then use the algorithms of \citeN{hickey-cohen} or \citeN{mckenzie} to sample uniformly from $L(G) \cap \Sigma^j$, and return the result.
The probability of obtaining any $w \in L(G)$ with $m \le |w| \le n$ is $(c_{|w|} / \sum_k c_k) \cdot (1 / |L(G) \cap \Sigma^{|w|}|) = 1 / \sum_k c_k = 1 / |\{ w \in L(G) \st m \le |w| \le n \}|$, as desired.
Furthermore, the counting and sampling both take expected time polynomial in $|G|$ and $n$.
\end{proof}

Putting this all together, we obtain a polynomial-time improvisation scheme for the case where the hard specification is a UCFG.
\begin{theorem} \label{theorem:ucfg-dfa-scheme}
There is a polynomial-time improvisation scheme for \cic{\UCFG}{\DFA}.
\end{theorem}
\begin{proof}
We follow the procedure of Theorem \ref{theorem:generic-scheme} with one modification: we do not perform length restriction explicitly, since the counting and sampling algorithms effectively do so (i.e. they only count/sample words of a given length).
Applying Lemma \ref{lemma:regular-intersection} to $\hard$ and $\soft$, we obtain a UCFG $\mathcal{\valids}$ such that $\{ w \in L(\mathcal{\valids}) \st m \le |w| \le n\} = \valids$.
Applying Lemma \ref{lemma:regular-intersection} to $\hard$ and the complement of $\soft$ (computed with the usual DFA construction), we obtain a UCFG $\mathcal{B}$ such that $\{ w \in L(\mathcal{B}) \st m \le |w| \le n \} = \improvs \setminus \valids$.
Next, applying Lemma \ref{lemma:ucfg-counting} to $\hard$ and $\mathcal{\valids}$, we obtain $|\improvs|$ and $|\valids|$.
Then we can proceed exactly as in Theorem \ref{theorem:generic-scheme}, using Lemma \ref{lemma:ucfg-sampling} to sample from $\{ w \in L(\mathcal{\valids}) \st m \le |w| \le n\} = \valids$ and $\{ w \in L(\mathcal{B}) \st m \le |w| \le n \} = \improvs \setminus \valids$.
\end{proof}

While context-free grammars are most likely to be useful as hard specifications, when used as soft specifications instead there is still a polynomial-time improvisation scheme.
However, the algorithm is significantly more involved than the one above, because we must sample from the complement of a context-free language and cannot directly use the sampling algorithm for UCFGs.
\begin{theorem}
There is a polynomial-time improvisation scheme for \cic{\DFA}{\UCFG}.
\end{theorem}
\begin{proof}
Applying length restriction as in Theorem \ref{theorem:dfa-scheme} to $\hard$, we obtain a DFA $\mathcal{\improvs}$ such that $L(\mathcal{\improvs}) = \improvs$.
Applying Lemma \ref{lemma:regular-intersection} to $\mathcal{\improvs}$ and $\soft$, we obtain a UCFG $\mathcal{\valids}$ such that $L(\mathcal{\valids}) = \valids$.
Then we can apply Lemma \ref{lemma:dfa-counting} to $\mathcal{\improvs}$ and Lemma \ref{lemma:ucfg-counting} to $\mathcal{\valids}$ to compute $|\improvs|$ and $|\valids|$.
If we can uniformly sample from $\Delta = L(\mathcal{\improvs}) \cap \overline{L(\mathcal{\valids})} = L(\mathcal{\improvs}) \setminus L(\mathcal{\valids}) = \improvs \setminus \valids$, we can then proceed exactly as in Theorem \ref{theorem:generic-scheme}.

We will build up a word from $\Delta$ incrementally, starting from the empty word.
Suppose we have generated the prefix $\sigma$ so far.
For each symbol $a \in \Sigma$, let $\Delta_{\sigma a} \subseteq \Delta$ contain all words starting with the prefix $\sigma a$, i.e. $\Delta_{\sigma a} = \{ w \in \Delta \st \exists z \in \Sigma^* : w = \sigma a z \}$.
Construct a DFA $P_{\sigma a}$ accepting all words that start with the prefix $\sigma a$ (clearly we can take $P_{\sigma a}$ to have size polynomial in $|\Sigma|$ and $n$).
Then $\Delta_{\sigma a} = \Delta \cap L(P_{\sigma a}) = \overline{\overline{\Delta} \cup \overline{L(P_{\sigma a})}} = \overline{\overline{L(\mathcal{\improvs})} \cup L(\mathcal{\valids}) \cup \overline{L(P_{\sigma a})}} = \Sigma^{\le n} \setminus (\overline{L(\mathcal{\improvs})} \cup L(\mathcal{\valids}) \cup \overline{L(P_{\sigma a})})$, since $L(\mathcal{\improvs}) \subseteq \Sigma^{\le n}$.
Letting $D$ be the complement of the product of $\mathcal{\improvs}$ and $P_{\sigma a}$, we have $L(D) = \overline{L(\mathcal{\improvs})} \cup \overline{L(P_{\sigma a})}$ and so $\Delta_{\sigma a} = \Sigma^{\le n} \setminus (L(\mathcal{\valids}) \cup L(D))$.
Let $D'$ be the product of $D$ and a DFA accepting all words of length at most $n$, so that $L(D') = L(D) \cap \Sigma^{\le n}$.
Applying Lemma \ref{lemma:regular-intersection} to $\mathcal{\valids}$ and $P_{\sigma a}$, we can find a UCFG $\mathcal{\valids}_{\sigma a}$ such that $L(\mathcal{\valids}_{\sigma a}) = L(\mathcal{\valids}) \cap L(P_{\sigma a})$.
Then we have
\begin{align*}
|\Delta_{\sigma a}| &= | \Sigma^{\le n} \setminus (L(\mathcal{\valids}) \cup L(D)) | \\
&= |\Sigma^{\le n}| - |\Sigma^{\le n} \cap (L(\mathcal{\valids}) \cup L(D))| \\
&= |\Sigma^{\le n}| - |L(\mathcal{\valids}) \cup L(D') | \\
&= |\Sigma^{\le n}| - |(L(\mathcal{\valids}) \cap \overline{L(D')}) \cup L(D') | \\
&= |\Sigma^{\le n}| - |L(\mathcal{\valids}) \cap \overline{L(D')}| -  |L(D') | \\
&= |\Sigma^{\le n}| - |L(\mathcal{\valids}) \cap (\overline{L(D)} \cup \Sigma^{> n})| -  |L(D') | \\
&= |\Sigma^{\le n}| - |L(\mathcal{\valids}) \cap L(\mathcal{\improvs}) \cap L(P_{\sigma a})| -  |L(D') | \\
&= |\Sigma^{\le n}| - |L(\mathcal{\valids}_{\sigma a})| -  |L(D')| .
\end{align*}
All words in $L(\mathcal{\valids}_{\sigma a})$ and $L(D')$ have length at most $n$, so using Lemmas \ref{lemma:ucfg-counting} and \ref{lemma:dfa-counting} respectively we can count the sizes of these languages in polynomial time.
Since $|\Sigma^{\le n}| = \sum_{i=0}^n |\Sigma|^i = (|\Sigma|^{n+1} - 1) / (|\Sigma| - 1)$, we can then compute $|\Delta_{\sigma a}|$ by the formula above.

Now let $\Delta_\sigma$ consist of all words in $\Delta$ that start with the prefix $\sigma$.
Observe that the sets $\Delta_{\sigma a}$ form a partition $\Pi_\sigma$ of $\Delta_\sigma$, unless $\sigma \in \Delta$ in which case we need to also add the set $\{\sigma\}$ to $\Pi_\sigma$.
Select one of the sets in $\Pi_\sigma$ randomly, with probability proportional to its size.
If we pick $\{\sigma\}$, then we stop and return $\sigma$ as our random sample.
Otherwise we picked $\Delta_{\sigma a}$ for some $a \in \Sigma$, so we append $a$ to $\sigma$ and continue.

Since the procedure of Theorem \ref{theorem:generic-scheme} only samples from $\improvs \setminus \valids = \Delta$ when it is nonempty, in the first iteration $\Delta_\sigma = \Delta$ is nonempty and so one of the sets in $\Pi_\sigma$ must be nonempty.
Since we assign a set in $\Pi_\sigma$ probability zero if it is empty, we will not select such a set, and so by induction $\Delta_\sigma$ is nonempty at every iteration.
We must terminate after at most $n$ iterations, since $\sigma$ gets longer in every iteration and all words in $\Delta$ have length at most $n$.
We prove by induction on $n - |\sigma|$ that if $\Delta_\sigma$ is nonempty, our procedure starting from $\sigma$ generates a uniform random sample from $\Delta_\sigma$.
In the base case $|\sigma| = n$, we have $\Delta_\sigma = \{\sigma\}$, since no word in $\Delta$ has length greater than $n$ but $\Delta_\sigma$ is nonempty.
Therefore $\Delta_{\sigma a} = \emptyset$ for all $a \in \Sigma$, so the procedure will return $\sigma$, and this is indeed a uniform random sample from $\Delta_\sigma$.
For $|\sigma| < n$, if $\sigma \in \Delta_\sigma$ then the probability of returning $\sigma$ is the probability of picking the set $\{\sigma\}$ from $\Pi_\sigma$, which is $1 / \sum_{S \in \Pi_\sigma} |S| = 1 / |\Delta_\sigma|$.
Any other $w \in \Delta_\sigma$ has length at least $|\sigma|+1$ and so can be written $w = \sigma a w'$ for some $a \in \Sigma$ and $w' \in \Sigma^*$.
Now to return $w$, our procedure must first pick the set $\Delta_{\sigma a}$ from $\Pi_\sigma$, and then return $w$ in a later iteration.
By the induction hypothesis, this happens with probability $(|\Delta_{\sigma a}| / \sum_{S \in \Pi_\sigma} |S|) \cdot (1 / |\Delta_{\sigma a}|) = 1 / |\Delta_\sigma|$.
So all words in $\Delta_\sigma$ are returned with uniform probability, and by induction this holds for all $\sigma$.
In particular it holds in the first iteration when $\sigma$ is the empty word, in which case we sample uniformly from $\Delta_\sigma = \Delta$ as desired.

This allows us to uniformly sample from $\improvs \setminus \valids$ and so finish the implementation of the procedure in Theorem \ref{theorem:generic-scheme}.
All of the operations we perform are polynomial-time, and the sampling procedure needs only polynomially-many iterations, so this yields a polynomial-time improvisation scheme.
\end{proof}

The results above show that there are polynomial-time improvisation schemes for CI instances where one specification is a UCFG and the other is a DFA.
However, when \emph{both} the hard and soft specifications are grammars (even unambiguous ones), the complexity jumps all the way up to $\sharpP$.
\begin{theorem} \label{theorem:2ucfg-hardness}
$\cic{\UCFG}{\UCFG}$ is $\sharpP$-hard.
\end{theorem}

To prove this, we introduce an intermediate problem that captures the essentially hard part of solving CI problems with two grammar specifications.
\begin{definition}
{\ucfgint} is the problem of computing, given two UCFGs $G_1$ and $G_2$ over an alphabet $\Sigma$ and $n \in \N$ in unary, the number of words in $L(G_1) \cap L(G_2) \cap \Sigma^{\le n}$.
\end{definition}

Without the length bound, checking emptiness of the intersection of two CFGs is undecidable, as shown by \citeN{bar-hillel}.
We use a similar proof (closer to that of \citeN{floyd-ambiguity}) to establish the complexity of the bounded version.

\begin{lemma} \label{lemma:ucfgint-hardness}
{\ucfgint} is $\sharpP$-complete.
\end{lemma}
\begin{proof}
All the words in $L(G_1) \cap L(G_2) \cap \Sigma^n$ have length polynomial in the size of the input, and checking whether a word is in the set can be done in polynomial time.
So $\ucfgint \in \sharpP$.

To show hardness, we give a reduction to {\ucfgint} from {\nbpcp}, the counting version of the bounded Post correspondence problem.
Recall that {\bpcp} instances are ordinary PCP instances together with a bound $m \in \N$ in unary on the total number of tiles that may be used in a solution.
In the usual construction simulating a Turing machine $M$ with PCP tiles (see for example \citeN{hmu}), an accepting computation history for $M$ leads to a PCP solution using a number of tiles linear in the size of the history.
So the problem of checking whether $M$ accepts a word in at most a number of steps given in unary reduces to {\bpcp}, and thus the latter is $\NP$-complete (as noted in a slightly different form by \citeN{constable-hunt-sahni}).
Furthermore, this reduction is almost parsimonious: the computation history uniquely determines the PCP solution until the accepting state is reached, when (at least in the reduction given in \citeN{hmu}) the accepting state can ``consume'' tape symbols on either side of it in any order.
However, we can easily modify the construction so that the symbols are consumed in a canonical order.
In the terminology of \citeN{hmu}, we replace the ``type-4'' tiles with tiles of the form $(qX, q)$ and $(Xq\#, q\#)$ for all accepting states $q$ and tape symbols $X$.
This forces tape symbols to be consumed from the right of $q$ first, and only allows consumption from the left once there are no tape symbols to the right (so that $q$ is adjacent to the $\#$ marking the end of the element of the history).
Now there is a one-to-one correspondence between accepting computation histories of $M$ and solutions to the {\bpcp} instance, so the counting problem {\nbpcp} is $\sharpP$-complete.

Given an instance $\mathcal{P}$ of ${\nbpcp}$ with tiles $(x_1, y_1), \dots, (x_k, y_k)$ over an alphabet $\Sigma'$ and a bound $m$, we construct grammars $X$ and $Y$ along the lines of \citeN{bar-hillel} and \citeN{floyd-ambiguity}.
The grammars use an alphabet $\Sigma$ consisting of $\Sigma'$ together with additional symbols $t_1, \dots, t_k$, and have nonterminal symbols $X_i$ and $Y_i$ respectively for $1 \le i \le m$.
Their start symbols are $X_m$ and $Y_m$ respectively, and they have the following productions:
\begin{align*}
X_i & \rightarrow x_1 X_{i-1} t_1 \st \cdots \st x_k X_{i-1} t_k \st X_1 & 2 \le i \le m \\
X_1 & \rightarrow x_1 t_1 \st \cdots \st x_k t_k \\
Y_i & \rightarrow y_1 Y_{i-1} t_1 \;\st \cdots \st \, y_k Y_{i-1} t_k \:\st Y_1 & 2 \le i \le m \\
Y_1 & \rightarrow y_1 t_1 \;|\; \cdots \st y_k t_k
\end{align*}
An easy induction shows that the words derivable from $X_i$ (respectively, $Y_i$) correspond to nonempty sequences of at most $i$ tiles.
Thus there is a one-to-one correspondence between words in $L(X) \cap L(Y)$ and solutions of $\mathcal{P}$.
Furthermore, both grammars are unambiguous, since the sequence of $t_i$ symbols uniquely determines the derivation tree.
Finally, all words in $L(X)$ have length at most $n = m (1 + \max(|x_1|, \dots, |x_k|))$, so $|L(X) \cap L(Y)| = |L(X) \cap L(Y) \cap \Sigma^{\le n}|$.
Therefore $(X, Y, 1^n)$ is a {\ucfgint} instance with the same number of solutions as $\mathcal{P}$.
This reduction clearly can be done in polynomial time, so {\ucfgint} is $\sharpP$-hard.
\end{proof}

Now we may complete the hardness proof for \cic{\UCFG}{\UCFG}.
\begin{proof}[of Theorem \ref{theorem:2ucfg-hardness}]
We reduce {\ucfgint} to checking feasibility of a \cic{\UCFG}{\UCFG} instance, along the same lines as Theorem \ref{theorem:nfa-hardness}.
Given a {\ucfgint} instance $(\hard, \soft, 1^n)$ and fixing some $N \in \N$, consider the CI instance $\mathcal{C}_N = (\hard, \soft, 0, n, 0, 0, 1/N)$.
For this instance we have $\improvs = L(\hard) \cap \Sigma^{\le n}$ and $\valids = L(\hard) \cap L(\soft) \cap \Sigma^{\le n}$.
By Theorem \ref{theorem:feasibility}, $\mathcal{C}_N$ is feasible iff $|\valids| \ge N$, so we can determine if $|L(\hard) \cap L(\soft) \cap \Sigma^{\le n}| \ge N$ with a feasibility query.
Since $|L(\hard) \cap L(\soft) \cap \Sigma^{\le n}| = O(|\Sigma|^n)$, we can compute $|L(\hard) \cap L(\soft) \cap \Sigma^{\le n}|$ with polynomially-many such queries by binary search.
This gives a polynomial-time (Cook) reduction from {\ucfgint} to $\cic{\UCFG}{\UCFG}$, so by Lemma \ref{lemma:ucfgint-hardness} the latter is $\sharpP$-hard.
\end{proof}

\section{Symbolic Specifications} \label{sec:symbolic}

In this section we discuss an important class of specifications given by existentially-quantified Boolean formulas.
We begin by defining the class and outline how it includes a widely-used succinct encoding of large finite-state automata.
Then we show that although the CI problem is $\sharpP$-hard with these specifications, it can be solved approximately using only an $\NP$ oracle.
This means that some CI problems with very large automata can still be solved in practice using {\SAT} solvers.

The specifications we will consider are Boolean formulas with auxiliary variables, where a word $w$ (encoded in binary) is in the language iff the formula is satisfiable after plugging in $w$.
\begin{definition}
Fixing an alphabet $\Sigma = \{0,1\}^k$ and a length $n \in \N$, a \emph{symbolic specification} is a Boolean formula $\phi(\overline{w}, \overline{a})$ where $\overline{w}$ is a vector of $nk$ variables and $\overline{a}$ is a vector of zero or more variables.
The language of $\phi$ consists of all $\overline{w} \in \Sigma^n$ such that $\exists \overline{a} \, \phi(\overline{w}, \overline{a})$ is true.
The class of symbolic specifications is denoted {\symb}.
\end{definition}

\begin{remark}
The restriction that all words in the language of any symbolic specification have a fixed length $n$ is for later convenience when intersecting two specifications.
We can always pad shorter words out to length $n$ with a dummy symbol added to the alphabet, altering the formula $\phi$ as appropriate.
\end{remark}

Symbolic specifications arise in the analysis of transition systems, and in practice can be useful even for DFAs in situations where there is insufficient memory to store full transition tables.
This is not uncommon in practice, because specifications are often built up as the conjunction of many small pieces.
To use our earlier improvisation scheme for DFAs, we would need to explicitly construct the product of these automata, which could be exponentially large.
An implicit representation of the product by Boolean formulas, in contrast, will have linear size.

In bounded model checking \cite{bmc}, DFAs and NFAs can be represented by formulas as follows:
\begin{definition}
A \emph{symbolic automaton} is a transition system over states $S \subseteq \{0,1\}^m$ and inputs $\Sigma \subseteq \{0,1\}^k$ represented by:
\begin{itemize}
\item a formula $\mathrm{init}(\overline{x})$ which is true iff $\overline{x} \in \{0,1\}^m$ is an initial state,
\item a formula $\mathrm{acc}(\overline{x})$ which is true iff $\overline{x} \in \{0,1\}^m$ is an accepting state, and
\item a formula $\delta(\overline{x}, \overline{c}, \overline{y})$ which is true iff there is a transition from $\overline{x} \in \{0,1\}^m$ to $\overline{y} \in \{0,1\}^m$ on input $\overline{c} \in \{0,1\}^k$ .
\end{itemize}
A symbolic automaton accepts words in $\Sigma^*$ according to the usual definition for NFAs.
\end{definition}

Given a symbolic automaton $\mathcal{D}$ and a length bound $n \in \N$, it is straightforward to generate a symbolic specification $\phi$ such that $L(\phi) = L(\mathcal{D}) \cap \Sigma^n$.
The auxiliary variables of $\phi$ encode accepting paths of length $n$ for a given word (see \citeN{bmc} for details), and existentially quantifying them yields a formula whose models correspond exactly to accepting words of length $n$.
So for the rest of the section we focus on symbolic specifications, but our results apply to symbolic automata in particular.

Since symbolic specifications can be arbitrary Boolean formulas, checking feasibility of CI problems involving them is obviously $\sharpP$-hard.
In the other direction, deciding membership of a word in a symbolic specification can be done with an $\NP$ oracle, so by Theorem \ref{theorem:sharpP-bound} there is a polynomial-time improvisation scheme relative to a $\sharpP$ oracle (indeed, this holds even if we extend the definition of symbolic specification to allow a larger but bounded number of quantifiers).
However, this scheme is not much use in practice, since $\sharpP$-complete problems are still very difficult to solve.
Instead, given the recent dramatic progress in {\SAT} solvers, one can ask what is possible using only an $\NP$ oracle.
As feasibility checking is $\sharpP$-hard, by Toda's theorem an $\NP$ oracle does not let us solve the CI problem exactly unless $\PH$ collapses.
However, we will show that an $\NP$ oracle does suffice to \emph{approximately} solve the CI problem with symbolic specifications.

When discussing approximate solutions to a CI instance $\mathcal{C}$, it will be useful to refer to improvisers for different values of $\epsilon$, $\lambda$, and $\rho$ than those required by $\mathcal{C}$.
Thus we may speak of an \emph{$(\epsilon, \lambda, \rho)$-improvising distribution}, \emph{$(\epsilon, \lambda, \rho)$-feasibility}, and so forth, and these have the obvious definitions.

Our method for solving CI problems with symbolic automata does not use the generic procedure of Theorem \ref{theorem:generic-scheme}, since it depends on comparing the sizes of $\improvs$ and $\valids$ and breaks down if these are only estimated (also, {\symb} is not closed under differences).
However, we still depend on techniques for uniform sampling, in this case approximate sampling of solutions to Boolean formulas using an {\NP} oracle \cite{jvv,bgp,unigen}.

\begin{lemma} \label{lemma:symbolic-samp}
There is a probabilistic algorithm with an {\NP} oracle that given a symbolic specification $\phi$ and any $\tau > 0$, returns a random sample from $L(\phi)$ that is uniform up to a factor of $1+\tau$ in expected time polynomial in $|\phi|$ and $1/\tau$ relative to the oracle.
\end{lemma}
\begin{proof}
By our definition above, words in the language of $\phi$ are assignments to the $\overline{w}$ variables of $\phi(\overline{w}, \overline{a})$ that can be extended to a complete satisfying assignment.
For an arbitrary formula $\phi$, counting the number of such words is exactly the problem {\NSAT} introduced by \citeN{valiant-sharpP} and also called \emph{projected model counting} (since it counts models projected onto a subset of variables).
Sampling from $L(\phi)$ is therefore equivalent to (approximate) uniform sampling
of the models of $\phi(\overline{w}, \overline{a})$ projected onto the variables $\overline{w}$.
This can be done with the algorithm of \citeN{unigen}, which will run in the time required relative to an {\NP} oracle.
Unfortunately for technical reasons the algorithm works only for $\tau > 6.84$.
For general $\tau$, we can modify the algorithm of \citeN{bgp} to do projection sampling (the modification is as described in \citeN{unigen}: apply the hash function only to the $\overline{w}$ variables, and when enumerating solutions using the {\NP} oracle block all solutions that agree on the $\overline{w}$ variables with the solutions already enumerated).
This algorithm has a constant probability of failing by returning $\bot$ instead of a sample.
If that happens we simply retry until the algorithm succeeds, and this only increases the expected runtime by a constant factor.
We note that the algorithm of \citeN{bgp} actually samples \emph{exactly} uniformly at random, and so would allow us to get somewhat better performance in theory, but with current {\SAT} solvers the algorithm does not scale (unlike the algorithm of \citeN{unigen}).
\end{proof}

Now we are ready to describe our approximate improvisation algorithm.
The algorithm avoids complementing $\soft$ (needed above to count and sample from $\improvs \setminus \valids$), and in fact does not require counting at all.
In exchange, for a desired $\epsilon$ the algorithm may not return an improviser achieving the best possible $\lambda$ and $\rho$, but the suboptimality is bounded.
\begin{theorem}
There is a procedure that given any $\tau > 0$ and a feasible CI problem $\mathcal{C} \in \cic{\symb}{\symb}$ with $\mathcal{C} = (\hard, \soft, n, n, \epsilon, \lambda, \rho)$, returns an $(\epsilon, \epsilon \lambda / (1+\tau), \rho (1+\epsilon)(1+\tau))$-improviser.
Furthermore, the procedure and the improvisers it generates run in expected time given by some fixed polynomial in $|\mathcal{C}|$ and $1/\tau$ relative to an {\NP} oracle.
\end{theorem}
\begin{proof}
Conjoin $\hard$ and $\soft$ (renaming any shared auxiliary variables) to produce a symbolic specification $\mathcal{\valids}$ such that $L(\mathcal{\valids}) = L(\hard) \cap L(\soft) = \valids$.
We build a probabilistic algorithm $G^{\NP}$ that approximately samples uniformly from $\valids$ with probability $1-\epsilon$ and from $\improvs$ with probability $\epsilon$.
By Lemma \ref{lemma:symbolic-samp}, we can do this with tolerance $\tau$ in polynomial expected time relative to the {\NP} oracle.

By definition, $G$ always returns an element of $I$, and returns an element of $A$ with probability at least $1-\epsilon$.
Since $\mathcal{C}$ is feasible, by Theorem \ref{theorem:feasibility} we have $1/\rho \le |\improvs| \le 1/\lambda$ and $(1-\epsilon)/\rho \le |\valids|$.
Now for any $w \in \valids$, we have $\Pr[w \leftarrow G] \ge (1-\epsilon)(1 / |\valids|) / (1+\tau) + \epsilon (1 / |\improvs|) / (1+\tau) \ge \lambda / (1+\tau)$, while for any $w \in \improvs \setminus \valids$ we have $\Pr[w \leftarrow G] \ge \epsilon (1 / |\improvs|) / (1 + \tau) \ge \epsilon \lambda / (1 + \tau)$.
Similarly, for any $w \in \valids$ we have $\Pr[w \leftarrow G] \le (1-\epsilon)(1 / |\valids|) (1+\tau) + \epsilon (1 / |\improvs|) (1+\tau) \le \rho (1 + \epsilon) (1+\tau)$, while for any $w \in \improvs \setminus \valids$ we have $\Pr[w \leftarrow G] \le \epsilon (1 / |\improvs|) (1+\tau) \le \epsilon \rho (1+\tau)$.
So $G$ is an $(\epsilon, \epsilon \lambda / (1+\tau), \rho (1+\epsilon) (1 + \tau))$-improviser.
\end{proof}

Therefore, it is possible to {\em approximately} solve the control
improvisation problem when the specifications are given by a succinct
propositional formula representation. This allows working with general
NFAs, and very large automata that cannot be stored explicitly, but
comes at the cost of using a {\SAT} solver (perhaps not a heavy cost
given the dramatic advances in the capacity of {\SAT} solvers) and
having to loosen the randomness requirement somewhat.

\section{Multiple Soft Constraints} \label{sec:multiple-soft}

In this section, we discuss a generalized problem, \emph{multi-constraint control improvisation (MCI)} where multiple soft constraints are allowed.
We introduced an earlier form of this extension in \citeN{akkaya-iotdi}, but did not theoretically analyze it.
Here we give a complete definition and a partial analysis.
First, the conditions under which an instance is feasible seem to be much more complex, and we do not have a simple form for them.
However, we can give an {\EXPTIME} algorithm to decide feasibility, in fact an exponential-time improvisation scheme.
Second, we show that even with DFA specifications, checking feasibility with multiple soft constraints is $\sharpP$-hard.

Observe that the definition of control improvisation we have used so far effectively permits multiple hard constraints: requiring that several specifications $\hard_1, \dots, \hard_k$ should all hold with probability 1 is equivalent to requiring that a single product specification $\hard$ hold with probability 1.
This property does not hold for soft constraints: requiring that both $\Pr[ w \in A_1 ] \ge 1/2$ and $\Pr[ w \in A_2 ] \ge 1/2$ is not in general equivalent to $\Pr[ w \in A ] \ge p$ for \emph{any} property $A$ and probability $p$.
Furthermore, there are situations such as the lighting control application of \citeN{akkaya-iotdi} where the natural formalization of the problem uses multiple soft constraints.
In this sense the asymmetry of the CI definition may be limiting.
 
The generalized definition we propose extends the basic one in a straightforward way: the soft specification $\soft$ and corresponding error probability $\epsilon$ are replaced by several specifications $\soft_1, \dots, \soft_k$ and error probabilities $\epsilon_1, \dots, \epsilon_k$.

\begin{definition} \label{defn:multi-improvs}
Fix a \emph{hard specification} $\hard$, \emph{soft specifications} $\soft_1, \dots, \soft_k$, and length bounds $m, n \in \N$.
An \emph{improvisation} is any word $w \in L(\hard)$ such that $m \le |w| \le n$, and we write $\improvs$ for the set of all improvisations.
An improvisation $w \in \improvs$ is $i$-\emph{admissible} if $w \in L(\soft_i)$, and we write $\valids_i$ for the set of all $i$-admissible improvisations.
\end{definition}

\begin{definition}
Given $\mathcal{C} = (\hard, \soft_1, \dots, \soft_k, m, n, \epsilon_1, \dots, \epsilon_k, \lambda, \rho)$
with $\hard$, $\soft_i$, $m$, and $n$ as in
Definition~\ref{defn:multi-improvs}, $\epsilon_i \in [0,1]
\cap \Q$ error probabilities, and $\lambda, \rho \in [0,1] \cap \Q$
probability 
bounds, a distribution $D : \Sigma^*
\rightarrow [0,1]$ is an
\emph{improvising distribution} if it satisfies the following requirements:
\begin{describe}{\emph{Hard constraint:}}
\item[\emph{Hard constraint:}] $\Pr [w \in \improvs \st w \leftarrow D ] = 1$
\item[\emph{Soft constraints:}] $\Pr[ w \in \valids_i \st w \leftarrow D ] \ge 1 - \epsilon_i \quad \forall i \in [k]$
\item[\emph{Randomness:}] $\forall w \in I$, $\; \lambda \le D(w) \le \rho$
\end{describe}
Feasibility, improvisers, and improvisation schemes are defined in terms of improvising distributions exactly as in Definitions \ref{defn:feasible} and \ref{defn:scheme}.
\end{definition}

\begin{definition}
Given $\mathcal{C} = (\hard, \soft_1, \dots, \soft_k, m, n, \epsilon_1, \dots, \epsilon_k, \lambda, \rho)$,
the \emph{multi-constraint control improvisation (MCI)} problem is to decide whether
$\mathcal{C}$ is feasible, and if so to generate an improviser for
$\mathcal{C}$.
The \emph{size} $|\mathcal{C}|$ of an MCI instance is measured as for CI instances.
\end{definition}

\begin{definition}
If $\mathcal{A}$ and $\mathcal{B}$ are classes of specifications, $\mcic{\mathcal{A}}{\mathcal{B}}$ is the class of MCI instances $\mathcal{C} = (\hard, \soft_1, \dots, \soft_k, m, n, \epsilon_1, \dots, \epsilon_k, \lambda, \rho)$ where $\hard \in \mathcal{A}$ and $\soft_i \in \mathcal{B}$ for all $i \in [k]$.
When discussing decision problems, we use the same notation for the feasibility problem associated with the class (i.e. given $\mathcal{C} \in \mcic{\mathcal{A}}{\mathcal{B}}$, decide whether it is feasible).
\end{definition}

\begin{remark}
Note that while the MCI definition treats the multiple soft constraints conjunctively (i.e. \emph{every} constraint must hold), if the type of specification used supports Boolean operations then other types of soft constraint can be brought to this form.
For example, the requirement $\Pr[w \in A_1] < 3/4 \implies \Pr[w \in A_2 \land w \in A_3] \ge 1/5$ can be rewritten $\Pr[w \in \overline{A_1}] \ge 1/4 \lor \Pr[w \in (A_2 \cap A_3)] \ge 1/5$, and then each disjunct tested for feasibility separately.
In general, we can reduce Boolean combinations of specifications inside probabilities to single specifications and write the resulting constraint in disjunctive normal form.
Each disjunct is then an MCI instance (ignoring the very minor issue of strict vs. non-strict inequalities), and the original problem with a complex soft constraint is feasible iff one of the disjuncts is.
This transformation could of course blow up the size of the problem exponentially --- the point is that slightly more complex soft constraints can be handled within the MCI framework.
\end{remark}

As was mentioned above, we do not know of a simple generalization of Theorem \ref{theorem:feasibility} giving necessary and sufficient conditions for the feasibility of an MCI instance.
It is straightforward to see that feasibility requires bounds not just on the sizes of the sets $\valids_i$ individually but also on the sizes of their intersections, pairwise, 3-wise, and so forth.
There are exponentially many such intersections, and so it is unclear whether there is a concise way to represent the needed conditions.
For a fixed $k$, however, we can directly write down the requirements on an improvising distribution as an exponentially-large linear program, and this gives an algorithm for feasibility and an improvisation scheme.

\begin{theorem}
If $\textsc{Spec}$ is a class of specifications for which membership can be decided in exponential time, then there is an exponential-time improvisation scheme for \mcic{\textsc{Spec}}{\textsc{Spec}}.
\end{theorem}
\begin{proof}
Say we are given an instance $\mathcal{C} \in \mcic{\textsc{Spec}}{\textsc{Spec}}$.
For any nonempty $M \subseteq [k]$, define $A'_M = \cap_{i \in M} A_i \setminus \cup_{i \not \in M} A_i$, i.e. the improvisations that are $i$-admissible for exactly those $i \in M$.
We also define $A'_\emptyset = \improvs \setminus \cup_i A_i$.
   Then the sets $A'_M$ for all $M \subseteq [k]$ are disjoint and partition $\improvs$, and $A_i = \cup_{M \ni i} A'_M$ for all $i \in [k]$.
   We will construct an improvising distribution by picking total probabilities $p_M$ for each set $A'_M$, and distributing this probability uniformly over the set.
 
 We build a linear program $\mathcal{P}$ over the variables $p_M$, $M \subseteq [k]$.
 In order to have a valid distribution, we require $p_M \ge 0$ and $\sum_{M} p_M \le 1$.
 Our distribution will automatically satisfy the hard constraint since $A'_M \subseteq \improvs$ for every $M \subseteq [k]$.
 To satisfy the soft constraints we require for each $i \in [k]$ that $\sum_{M \ni i} p_M \ge 1 - \epsilon_i$.
  Finally, since by the randomness requirement every element of $\improvs$ must have probability between $\lambda$ and $\rho$, we also require $\lambda |A'_M| \le p_M \le \rho |A'_M|$.
  
Suppose an improvising distribution $D$ exists, and assign $p_M = \Pr[ w \in A'_M \st w \leftarrow D]$.
Obviously $0 \le p_M \le 1$.
Since the sets $A'_M$ are disjoint, for any $i \in [k]$ we have $\sum_{M \ni i} p_M = \sum_{M \ni i} \Pr[ w \in A'_M \st w \leftarrow D] = \Pr[ w \in \cup_{M \ni i} A'_M \st w \leftarrow D] = \Pr[ w \in A_i \st w \leftarrow D] \ge 1 - \epsilon_i$.
Finally, $\lambda |A'_M| = \sum_{w \in A'_M} \lambda \le \sum_{w \in A'_M} D(w) = \Pr[w \in A'_M \st w \leftarrow D] = p_M$, and similarly $\rho |A'_M| \ge p_M$.
So our program $\mathcal{P}$ is feasible.
Conversely, given a solution to $\mathcal{P}$, let $D$ be the distribution which assigns $w \in \improvs$ probability $p_M / |A'_M|$ for the unique $M \subseteq[k]$ such that $w \in A'_M$ (recalling the sets $A'_M$ partition $\improvs$).
Since $p_M \ge 0$ for all $M \subseteq [k]$ and $\sum_{w \in \improvs} D(w) = \sum_{M \subseteq [k]} \sum_{w \in A'_M} p_M / |A'_M| = \sum_{M \subseteq [k]} p_M \le 1$, this is a valid distribution.
Obviously $D$ satisfies the hard constraint requirement, and for every $i \in [k]$ we have $\Pr[ w \in A_i \st w \leftarrow D] = \sum_{w \in A_i} D(w) = \sum_{M \ni i} \sum_{w \in A'_M} p_M / |A'_M| = \sum_{M \ni i} p_M \ge 1 - \epsilon_i$, so $D$ satisfies the soft constraint requirement.
Finally, for the $M \subseteq [k]$ such that $w \in A'_M$ we have $D(w) = p_M / |A'_M| \ge (\lambda |A'_M|) / |A'_M| = \lambda$ and $D(w) = p_M / |A'_M| \le (\rho |A'_M|) / |A'_M| = \rho$, so $D$ satisfies the randomness requirement.
Therefore $\mathcal{C}$ is feasible iff $\mathcal{P}$ is feasible.

 Since linear programming is polynomial-time, we can solve the exponentially-large program $\mathcal{P}$ in exponential time.
 If it is infeasible, then so is $\mathcal{C}$ and we return $\bot$.
 Otherwise, enumerate every $w \in \improvs = \{ w \in L(\hard) \st m \le |w| \le n \}$ and check for each one which of the specifications $\soft_i$ it satisfies, thereby determining the unique $M(w) \subseteq [k]$ such that $w \in A'_{M(w)}$.
 In the process keep track of how many words are in each $A'_M$.
 This can be done in exponential time, since there are exponentially-many words in $\Sigma^{\le n}$ and checking each one against $\hard$ and every $\soft_i$ takes exponential time.
 Generating $w \in \improvs$ with probability $p_{M(w)} / |A'_{M(w)}|$ yields the distribution $D$ above, which is an improvising distribution for $\mathcal{C}$.
\end{proof}

As evidence that the CI problem becomes much harder when there are multiple soft constraints, we consider DFA specifications.
Recall that by Theorem \ref{theorem:dfa-scheme}, there is a polynomial-time improvisation scheme for \cic{\DFA}{\DFA}.
This is unlikely to be the case for the corresponding multi-constraint problems.
\begin{theorem}
\mcic{\DFA}{\DFA} is \sharpP-hard.
\end{theorem}
\begin{proof}
We give a reduction from {\sharpSAT} to \mcic{\DFA}{\DFA} similar to that used by \citeN{sharpNFA} to prove hardness of counting the language of an NFA.
Suppose we are given a formula $F$ in conjunctive normal form, i.e. $F = c_1 \land \dots \land c_k$ where each $c_i$ is a disjunction of variables and their negations.
If $F$ has $n$ variables $V$, putting them in some order we can view assignments to $V$ as words in $\Sigma^n$ where $\Sigma = \{0,1\}$.
Then for each $c_i$ we can build a DFA $D_i$ accepting assignments that satisfy $c_i$.
We have one state for each $v \in V$, and start from the first variable under the order.
If in state $v$ we read a $1$ and $v$ occurs positively in $c_i$, or we read a $0$ and $v$ occurs negatively, then the assignment satisfies $c_i$ and we transition to a chain of states that ensure we accept iff the word has length exactly $n$.
Otherwise $c_i$ is not yet satisfied, so we move to the state for the next variable in the order.
If we are already at the last variable, then $c_i$ is not satisfied by the assignment and we reject.
Clearly, each $D_i$ has size polynomial in $|F|$, and the intersection $\cap_i L(D_i)$ contains precisely the satisfying assignments of $F$.

Now we proceed along similar lines to Theorem \ref{theorem:nfa-hardness}.
For any $N \in \N$, consider the MCI instance $\mathcal{C}_N = (\mathcal{T}, D_1, \dots, D_k, n, n, 0, \dots, 0, 0, 1/N)$ where $\mathcal{T}$ is the trivial DFA accepting all of $\Sigma^*$.
For this instance we have $\improvs = \Sigma^n$ and $\valids_i = L(D_i)$ for every $i \in \{1, \dots, k\}$.
Since $\epsilon_i = 0$ for every $i$, an improvising distribution for $\mathcal{C}_N$ must assign probability zero to all words not in $\cap_i L(D_i)$.
Therefore since no word can be assigned probability greater than $\rho$, if an improvising distribution exists then $|\cap_i L(D_i)| \ge 1/\rho = N$.
Conversely, if $|\cap_i L(D_i)| \ge N$ then a uniform distribution on $\cap_i L(D_i)$ is clearly an improvising distribution for $\mathcal{C}_N$.
So $\mathcal{C}_N$ is feasible iff $|\cap_i L(D_i)| \ge N$.
Since $|\cap_i L(D_i)| \le |\Sigma|^n = 2^n$, by binary search we can determine $|\cap_i L(D_i)|$ with polynomially-many \mcic{\DFA}{\DFA} queries, and thereby count the satisfying assignments of $F$.
\end{proof}

Note that there is a gap between our upper and lower bounds for \mcic{\DFA}{\DFA}: it is $\sharpP$-hard, but our earlier algorithm puts it only in $\EXPTIME$.
A $\P^\sharpP$ algorithm, or more plausibly a $\PSPACE$ algorithm, would be very interesting.
A multi-constraint version of Theorem \ref{theorem:feasibility} giving relatively simple conditions for feasibility would be helpful here.

\section{Conclusion}

In this paper, we introduced control improvisation, the problem of
creating improvisers that randomly generate words subject to hard and soft specifications.
We gave precise conditions for when improvisers
exist, and investigated the complexity of finding improvisers for
several major classes of automata and grammars.
In particular, we showed that the control improvisation problem for DFAs, as well as for DFAs and unambiguous CFGs, can be solved in polynomial time.
For NFAs and general CFGs, on the other hand, polynomial-time improvisation schemes are unlikely to exist, and would imply $\P = \P^\sharpP$.
These results are summarized in Table \ref{table:complexities}.
We also studied the case where the
specifications are presented symbolically instead of as explicit automata, and showed
that the control improvisation problem can still be solved
approximately using SAT solvers.
Finally, we explored a generalized form of the problem allowing multiple soft constraints, and gave evidence suggesting that it is substantially more difficult than the basic problem: even for DFA specifications it is $\sharpP$-hard, and we give only an $\EXPTIME$ algorithm.

\begin{table}[tp]
\tbl{Complexity of the control improvisation problem when $\hard$ and $\soft$ are various different types of specifications. Here $\sharpP$ indicates that feasibility checking is $\sharpP$-hard, and that there is a polynomial-time improvisation scheme relative to a $\sharpP$ oracle.\label{table:complexities}}{
\setlength{\tabcolsep}{5pt}
\renewcommand{\arraystretch}{1.2}
\begin{tabular}{cr||c|c|c|c|}
 & $\soft$ & \textbf{DFA} & \multicolumn{2}{c|}{\textbf{CFG}} & \textbf{NFA} \\
$\hard$ &  &  & unamb. & amb. &  \\
\hline
\hline
\textbf{DFA} &  & poly-time & poly-time & \multirow{3}{*}{} & \\
\cline{1-4}
\multirow{2}{*}{\textbf{CFG}} & unambiguous & poly-time & $\sharpP$ & &  \\
\cline{2-4}
   & ambiguous & \multicolumn{3}{r|}{\sharpP} &  \\
\cline{1-5}
\textbf{NFA} &  & \multicolumn{4}{r|}{\sharpP}  \\
\hline
\end{tabular}
}
\end{table}

One interesting direction for future
work would be to find other tractable cases of the control
improvisation problem deriving from finer structural properties of the
automata than just determinism.
There are also several clear directions for extending the problem definition other than allowing multiple soft constraints.
For example, it would be useful in robotics applications for improvisations to be infinite words, with specifications given in linear temporal logic.
We are also studying \emph{online} or \emph{reactive} extensions where improvisations are generated incrementally in response to a possibly adversarial environment: this is useful in musical applications as well as in robotics.
A third type of extension would allow soft specifications to be different kinds of quantitative constraints.
In a robotics application, for example, we might want to bound the average length of generated trajectories or the density with which they cover some region.
Finally, we are also investigating
further applications of control improvisation, particularly in the areas of testing, security,
and privacy.

\appendix
\section*{APPENDIX}
\setcounter{section}{1}

For completeness, we prove results on counting and sampling from the language of a DFA in the form we need.
The technique is classical, essentially that of \citeN{hickey-cohen}.

\renewcommand\thesection{4}
\setcounter{theorem}{1}

\begin{lemma} 
If $\mathcal{D}$ is a DFA, $|L(\mathcal{D})|$ can be computed in polynomial time.
\end{lemma}
\begin{proof}
First we prune irrelevant states unreachable from the initial state or from which no accepting state can be reached (this pruning can clearly be done in polynomial time). If the resulting graph contains a cycle (also detectable in polynomial time), we return $\infty$. Otherwise $\mathcal{D}$ is a DAG with multiple edges, and every sink is an accepting state. For each accepting state $s$ we add a new vertex and an edge to it from $s$. Then there is a one-to-one correspondence between accepting words of $\mathcal{D}$ and paths from the initial state to a sink. Now we can compute for each vertex $v$ the number of paths $p_v$ from it to a sink using the usual linear-time DAG algorithm (traversal in reverse topological order) modified slightly to handle multiple edges. We return $p_v$ with $v$ the initial state.
\end{proof}

\begin{lemma} 
There is a probabilistic algorithm that given a DFA $\mathcal{D}$ with finite language, returns a uniform random sample from $L(\mathcal{D})$ in polynomial expected time.
\end{lemma}
\begin{proof}
Prune $\mathcal{D}$ and compute the path counts $p_v$ as in Lemma \ref{lemma:dfa-counting} (since $\mathcal{D}$ has finite language, after pruning it is a DAG, as above). To every edge $(u,v)$ in $\mathcal{D}$ assign the weight $p_v / p_u$. It is clear that at every vertex the sum of the weights of the outgoing edges is 1 (unless the vertex is a sink). We prove by induction along reverse topological order that treating these weights as transition probabilities, starting from any state $u$ and talking a random walk until a sink is reached we obtain a uniform distribution over all paths from $u$ to a sink. If $u$ is a sink this holds trivially. If $u$ has a nonempty set of children $S$, then by the inductive hypothesis for every $v \in S$ starting a walk at $v$ gives a uniform distribution over the $p_v$ paths from $v$ to a sink. Therefore the probability of following any such path starting at $u$ is $(p_v/p_u) \cdot (1/p_v) = 1/p_u$. So the result holds by induction. In particular, if we start from the initial state we obtain a uniform distribution over all paths to a sink, and thus a uniform distribution over $L(\mathcal{D})$. Since all probabilities are rational with denominators bounded by ${|\Sigma|}^{|\mathcal{D}|}$, this walk can be performed by a probabilistic algorithm $S$ of size polynomial in $|\mathcal{D}|$, with expected time bounded by a fixed polynomial in $|\mathcal{D}|$. Then $S$ returns a uniform sample from $L(\mathcal{D})$, and it can be constructed in time polynomial in $|\mathcal{D}|$.
\end{proof}

\begin{acks}
The authors dedicate this paper to the memory of David Wessel, who was instrumental in the first work on control improvisation but passed away while the earlier version of this paper (\citeN{fsttcs-version}) was being written.
We would also like to thank Ben Caulfield, Orna Kupferman, and Markus Rabe for their helpful comments.
\end{acks}

\bibliographystyle{ACM-Reference-Format-Journals}
\bibliography{main.bib}


\begin{thebibliography}{00}


\ifx \showCODEN    \undefined \def \showCODEN     #1{\unskip}     \fi
\ifx \showDOI      \undefined \def \showDOI       #1{{\tt DOI:}\penalty0{#1}\ }
  \fi
\ifx \showISBNx    \undefined \def \showISBNx     #1{\unskip}     \fi
\ifx \showISBNxiii \undefined \def \showISBNxiii  #1{\unskip}     \fi
\ifx \showISSN     \undefined \def \showISSN      #1{\unskip}     \fi
\ifx \showLCCN     \undefined \def \showLCCN      #1{\unskip}     \fi
\ifx \shownote     \undefined \def \shownote      #1{#1}          \fi
\ifx \showarticletitle \undefined \def \showarticletitle #1{#1}   \fi
\ifx \showURL      \undefined \def \showURL       #1{#1}          \fi

\bibitem[\protect\citeauthoryear{Akkaya, Fremont, Valle, Donz{\'e}, Lee, and
  Seshia}{Akkaya et~al\mbox{.}}{2016}]%
        {akkaya-iotdi}
{Ilge Akkaya}, {Daniel~J. Fremont}, {Rafael Valle}, {Alexandre Donz{\'e}},
  {Edward~A. Lee}, {and} {Sanjit~A. Seshia}. 2016.
\newblock \showarticletitle{Control Improvisation with Probabilistic Temporal
  Specifications}. In {\em 1st IEEE Conference on Internet-of-Things Design and
  Implementation (IoTDI)}.
\newblock


\bibitem[\protect\citeauthoryear{Arora and Barak}{Arora and Barak}{2009}]%
        {arora-barak}
{Sanjeev Arora} {and} {Boaz Barak}. 2009.
\newblock {\em Computational Complexity: A Modern Approach}.
\newblock Cambridge University Press, New York.
\newblock


\bibitem[\protect\citeauthoryear{Assayag, Bloch, Chemillier, L{\'e}vy, and
  Dubnov}{Assayag et~al\mbox{.}}{2012}]%
        {omax}
{G{\'e}rard Assayag}, {Georges Bloch}, {Marc Chemillier}, {Benjamin L{\'e}vy},
  {and} {Shlomo Dubnov}. 2012.
\newblock {OMax}.
\newblock \url{http://repmus.ircam.fr/omax/home}.   (2012).
\newblock


\bibitem[\protect\citeauthoryear{Assayag and Dubnov}{Assayag and
  Dubnov}{2004}]%
        {AssayagD04}
{G{\'e}rard Assayag} {and} {Shlomo Dubnov}. 2004.
\newblock \showarticletitle{Using Factor Oracles for Machine Improvisation}.
\newblock {\em Soft Comput.\/} {8}, 9 (2004), 604--610.
\newblock


\bibitem[\protect\citeauthoryear{Bar-Hillel, Perles, and Shamir}{Bar-Hillel
  et~al\mbox{.}}{1961}]%
        {bar-hillel}
{Yehoshua Bar-Hillel}, {Micha Perles}, {and} {Eliyahu Shamir}. 1961.
\newblock \showarticletitle{On Formal Properties of Simple Phrase Structure
  Grammars}.
\newblock {\em Zeitschrift f{\"u}r Phonetik, Sprachwissenschaft und
  Kommunikationsforschung\/}  {14} (1961), 143--172.
\newblock
\newblock
\shownote{Reprinted in Y. Bar-Hillel, {\it Language and Information},
  Addison-Wesley, 1964.}


\bibitem[\protect\citeauthoryear{Bellare, Goldreich, and Petrank}{Bellare
  et~al\mbox{.}}{2000}]%
        {bgp}
{Mihir Bellare}, {Oded Goldreich}, {and} {Erez Petrank}. 2000.
\newblock \showarticletitle{Uniform generation of {NP}-witnesses using an
  {NP}-oracle}.
\newblock {\em Information and Computation\/} {163}, 2 (2000), 510--526.
\newblock


\bibitem[\protect\citeauthoryear{Biere, Cimatti, Clarke, and Zhu}{Biere
  et~al\mbox{.}}{1999}]%
        {bmc}
{Armin Biere}, {Alessandro Cimatti}, {Edmund~M. Clarke}, {and} {Yunshan Zhu}.
  1999.
\newblock \showarticletitle{Symbolic Model Checking without {BDD}s}. In {\em
  5th International Conference on Tools and Algorithms for Construction and
  Analysis of Systems}. 193--207.
\newblock


\bibitem[\protect\citeauthoryear{Cassandras and Lafortune}{Cassandras and
  Lafortune}{2006}]%
        {lafortune06}
{Christos~G. Cassandras} {and} {St{\'e}phane Lafortune}. 2006.
\newblock {\em Introduction to Discrete Event Systems}.
\newblock Springer-Verlag New York, Inc., Secaucus, NJ, USA.
\newblock
\showISBNx{0387333320}


\bibitem[\protect\citeauthoryear{Chakraborty, Meel, and Vardi}{Chakraborty
  et~al\mbox{.}}{2014}]%
        {unigen}
{Supratik Chakraborty}, {Kuldeep~S. Meel}, {and} {Moshe~Y. Vardi}. 2014.
\newblock \showarticletitle{Balancing Scalability and Uniformity in {SAT}
  Witness Generator}. In {\em 51st Design Automation Conference}. ACM, 1--6.
\newblock


\bibitem[\protect\citeauthoryear{Cleophas, Zwaan, and Watson}{Cleophas
  et~al\mbox{.}}{2003}]%
        {Cleophas03constructingfactor}
{Loek Cleophas}, {Gerard Zwaan}, {and} {Bruce~W. Watson}. 2003.
\newblock \showarticletitle{Constructing Factor Oracles}. In {\em In
  Proceedings of the 3rd Prague Stringology Conference}.
\newblock


\bibitem[\protect\citeauthoryear{Constable, Hunt~III, and Sahni}{Constable
  et~al\mbox{.}}{1974}]%
        {constable-hunt-sahni}
{Robert~L. Constable}, {Harry~B. Hunt~III}, {and} {Sartaj Sahni}. 1974.
\newblock {\em On the computational complexity of scheme equivalence}.
\newblock {T}echnical {R}eport TR 74-201. Cornell University.
\newblock
\showURL{%
\url{http://hdl.handle.net/1813/6041}}


\bibitem[\protect\citeauthoryear{Donz{\'{e}}, Valle, Akkaya, Libkind, Seshia,
  and Wessel}{Donz{\'{e}} et~al\mbox{.}}{2014}]%
        {donze-icmc14}
{Alexandre Donz{\'{e}}}, {Rafael Valle}, {Ilge Akkaya}, {Sophie Libkind},
  {Sanjit~A. Seshia}, {and} {David Wessel}. 2014.
\newblock \showarticletitle{Machine Improvisation with Formal Specifications}.
  In {\em Proceedings of the 40th International Computer Music Conference
  (ICMC)}.
\newblock


\bibitem[\protect\citeauthoryear{Floyd}{Floyd}{1962}]%
        {floyd-ambiguity}
{Robert~W. Floyd}. 1962.
\newblock \showarticletitle{On ambiguity in phrase structure languages}.
\newblock {\it Commun. ACM} {5}, 10 (1962), 526.
\newblock


\bibitem[\protect\citeauthoryear{Fremont, Donz{\'e}, Seshia, and
  Wessel}{Fremont et~al\mbox{.}}{2015}]%
        {fsttcs-version}
{Daniel~J. Fremont}, {Alexandre Donz{\'e}}, {Sanjit~A. Seshia}, {and} {David
  Wessel}. 2015.
\newblock \showarticletitle{{Control Improvisation}}. In {\em 35th IARCS Annual
  Conference on Foundations of Software Technology and Theoretical Computer
  Science (FSTTCS 2015)}. 463--474.
\newblock


\bibitem[\protect\citeauthoryear{Ginsburg and Ullian}{Ginsburg and
  Ullian}{1966}]%
        {ginsburg-ullian}
{Seymour Ginsburg} {and} {Joseph Ullian}. 1966.
\newblock \showarticletitle{Ambiguity in context free languages}.
\newblock {\em Journal of the ACM (JACM)\/} {13}, 1 (1966), 62--89.
\newblock


\bibitem[\protect\citeauthoryear{Gordon, Henzinger, Nori, and Rajamani}{Gordon
  et~al\mbox{.}}{2014}]%
        {probprog}
{Andrew~D. Gordon}, {Thomas~A. Henzinger}, {Aditya~V. Nori}, {and} {Sriram~K.
  Rajamani}. 2014.
\newblock \showarticletitle{Probabilistic Programming}. In {\em International
  Conference on Software Engineering (ICSE Future of Software Engineering)}.
  IEEE.
\newblock
\showURL{%
\url{http://research.microsoft.com/apps/pubs/default.aspx?id=208585}}


\bibitem[\protect\citeauthoryear{Hickey and Cohen}{Hickey and Cohen}{1983}]%
        {hickey-cohen}
{Timothy Hickey} {and} {Jacques Cohen}. 1983.
\newblock \showarticletitle{Uniform random generation of strings in a
  context-free language}.
\newblock {\it SIAM J. Comput.} {12}, 4 (1983), 645--655.
\newblock


\bibitem[\protect\citeauthoryear{Hopcroft, Motwani, and Ullman}{Hopcroft
  et~al\mbox{.}}{2001}]%
        {hmu}
{John~E. Hopcroft}, {Rajeev Motwani}, {and} {Jeffrey~D. Ullman}. 2001.
\newblock {\em Introduction to Automata Theory, Languages, and Computation\/}
  (2nd ed.).
\newblock Addison-Wesley.
\newblock


\bibitem[\protect\citeauthoryear{Jerrum, Valiant, and Vazirani}{Jerrum
  et~al\mbox{.}}{1986}]%
        {jvv}
{Mark~R. Jerrum}, {Leslie~G. Valiant}, {and} {Vijay~V. Vazirani}. 1986.
\newblock \showarticletitle{Random generation of combinatorial structures from
  a uniform distribution}.
\newblock {\em Theoretical Computer Science\/}  {43} (1986), 169--188.
\newblock


\bibitem[\protect\citeauthoryear{Kannan, Sweedyk, and Mahaney}{Kannan
  et~al\mbox{.}}{1995}]%
        {sharpNFA}
{Sampath Kannan}, {Z. Sweedyk}, {and} {Steve Mahaney}. 1995.
\newblock \showarticletitle{Counting and random generation of strings in
  regular languages}. In {\em Sixth Annual ACM-SIAM Symposium on Discrete
  Algorithms}. SIAM, 551--557.
\newblock


\bibitem[\protect\citeauthoryear{Lafortune}{Lafortune}{2015}]%
        {lafortune-personal15}
{St{\'e}phane Lafortune}. 2015.
\newblock Personal Communication.   (2015).
\newblock


\bibitem[\protect\citeauthoryear{McKenzie}{McKenzie}{1997}]%
        {mckenzie}
{Bruce McKenzie}. 1997.
\newblock {\em Generating strings at random from a context free grammar}.
\newblock Technical Report TR-COSC 10/97. University of Canterbury.
\newblock


\bibitem[\protect\citeauthoryear{Rabin}{Rabin}{1963}]%
        {rabin-pfas}
{Michael~O. Rabin}. 1963.
\newblock \showarticletitle{Probabilistic Automata}.
\newblock {\em Information and Control\/} {6}, 3 (1963), 230--245.
\newblock


\bibitem[\protect\citeauthoryear{Rowe}{Rowe}{2001}]%
        {rowe-2001}
{Robert Rowe}. 2001.
\newblock {\em Machine Musicianship}.
\newblock MIT Press.
\newblock


\bibitem[\protect\citeauthoryear{Sutton, Greene, and Amini}{Sutton
  et~al\mbox{.}}{2007}]%
        {fuzzing-book}
{Michael Sutton}, {Adam Greene}, {and} {Pedram Amini}. 2007.
\newblock {\em Fuzzing: Brute Force Vulnerability Discovery}.
\newblock Addison-Wesley.
\newblock


\bibitem[\protect\citeauthoryear{Toda}{Toda}{1991}]%
        {toda}
{Seinosuke Toda}. 1991.
\newblock \showarticletitle{{PP} is as hard as the polynomial-time hierarchy}.
\newblock {\it SIAM J. Comput.} {20}, 5 (1991), 865--877.
\newblock


\bibitem[\protect\citeauthoryear{Valiant}{Valiant}{1979}]%
        {valiant-sharpP}
{Leslie~G. Valiant}. 1979.
\newblock \showarticletitle{The complexity of computing the permanent}.
\newblock {\em Theoretical Computer Science\/} {8}, 2 (1979), 189--201.
\newblock


\end{thebibliography}

\received{Month Year}{Month Year}{Month Year}

\end{document}